\documentclass[12pt]{article}
\usepackage[margin=1in]{geometry}

\usepackage{graphicx,verbatim,array,multicol, subfigure, xcolor, lscape, mathrsfs}
\usepackage{psfrag, amsmath, amsfonts, epsfig, fancybox, setspace,soul,amsthm,natbib,pdflscape}
\usepackage{hyperref}
\def\supg{^{(g)}}
\def\bW{{\bf W}}

\def\bL{{\bf L}}
\def\bw{{\bf w}}
\def\supone{^{(1)}}
\def\supzero{^{(0)}}
\def\supg{^{(g)}}
\def\transpose{{\sf \scriptscriptstyle{T}}}

\def\trans{^{\transpose}}
\def\var{\text{Var}}

\def\calt{\mathcal{T}}
\def\calp{\mathcal{P}}

\newtheorem{subtheorem}{Theorem}

\newtheorem{lemma}{Lemma}

\usepackage{colortbl}
\usepackage{xcolor}
\usepackage{setspace}
\usepackage{caption}

\begin{document}

\begin{center}
\thispagestyle{empty}
\LARGE{Efficient Testing Using Surrogate Information}\\

\vspace*{15mm}
\normalsize Rebecca Knowlton$^{1}$ and Layla Parast$^{2}$ \\
\vspace*{10mm}
\normalsize $^{1}$Department of Statistics and Data Sciences, University of Texas at Austin,\\  \normalsize 105 E 24th St D9800, Austin, TX 78705, rknowlton@utexas.edu \\
\normalsize $^{2}$Department of Statistics and Data Sciences, University  of Texas at Austin, \\ \normalsize 105 E 24th St D9800, Austin, TX 78705, parast@austin.utexas.edu
\end{center}

\vfill 
\noindent Acknowledgments: This work was supported by NIDDK grant R01DK118354 (PI:Parast). We are grateful to the AIDS Clinical Trial Group (ACTG) Network for supplying the HIV clinical trial data.

\clearpage
\thispagestyle{empty}
\begin{abstract}

In modern clinical trials, there is immense pressure to use surrogate markers in place of an expensive or long-term primary outcome to make more timely decisions about treatment effectiveness.  However, using a surrogate marker to test for a treatment effect can be difficult and controversial. Existing methods tend to either rely on fully parametric methods where strict assumptions are made about the relationship between the surrogate and the outcome, or assume the surrogate marker is valid for the entire study population. In this paper, we develop a fully nonparametric method for efficient testing using surrogate information (ETSI). Our approach is specifically designed for settings where there is heterogeneity in the utility of the surrogate marker, i.e., the surrogate is valid for certain patient subgroups and not others. ETSI enables treatment effect estimation and hypothesis testing via kernel-based estimation for a setting where the surrogate is used in place of the primary outcome for individuals for whom the surrogate is valid, and the primary outcome is purposefully only measured in the remaining patients. In addition, we provide a framework for future study design with power and sample size estimates based on our proposed testing procedure. We demonstrate the performance of our methods via a simulation study and application to two distinct HIV clinical trials.

\end{abstract}


\clearpage
\setcounter{page}{1}
\section{Introduction}
\allowdisplaybreaks
A common challenge in clinical studies occurs when measuring the primary outcome of interest requires long patient follow-up or is otherwise expensive, typically in terms of patient burden or financial constraints. To alleviate these costs, it has become common practice for researchers to consider using a surrogate marker. A surrogate marker is a measure that is more easily (or quickly) obtained than the primary outcome and that, once validated, could be used to replace the primary outcome when evaluating the effectiveness of a treatment. For complex diseases such as HIV/AIDS and cancer, the use of surrogate markers has greatly improved the ability to evaluate treatment effectiveness quickly and efficiently \citep{fleming1994, degruttola1997, katz2004}. Importantly, a surrogate marker must first be validated to ensure that the treatment effect on the surrogate faithfully translates to the treatment effect on the primary outcome \citep{prentice1989}. Several useful statistical frameworks have been proposed for this purpose \citep{elliott2023surrogate,freedman1992, wang2002, frangakis2002, buyse2000}.

Once a surrogate marker is validated, the ultimate goal is to make a conclusion about the effectiveness of a treatment in a future trial, using the surrogate marker \textit{instead of} the primary outcome. Note that this is different from trying to use a surrogate marker to gain statistical efficiency when estimating a treatment effect. This latter problem of using a surrogate, or auxiliary data in general, to gain efficiency is a well-studied problem. For example, \cite{pepe1992inference} proposes a useful likelihood-based method to gain efficiency in a setting where surrogate information is available for all patients, but the primary outcome is only available for a random subset of patients.  \cite{leung2001} provides a review on so-called \textit{augmented} surrogate studies, which are, similar to \cite{pepe1992inference}, specific to the setting where surrogate information is available for all individuals and the outcome is available for a randomly selected subset. Within a mediation framework, \cite{zhou2023} propose methods to gain efficiency and power by incorporating low and high dimensional mediator information via a parametric linear model.

In this paper, we are not focused on using the surrogate marker to gain statistical efficiency. Instead, we are interested in using the surrogate in place of the primary outcome to test for a treatment effect in a future study. Methods that do address this question tend to either rely on fully parametric methods where strict parametric assumptions are made about the relationship between the surrogate and the outcome, or assume the surrogate marker is valid for the entire study population.  For example, \citet{price2018estimation} define an optimal surrogate which aims to predict the primary outcome with estimation via the super-learner and targeted super-learner, and propose to test for a treatment effect using this optimal surrogate. Under the assumption that the true treatment effects on the surrogate and the primary outcome are bivariate normal, \citet{quan2023utilization} and \citet{saint2019predictive} propose methods that use prior information about the treatment effect on a surrogate to plan a future study. As a nonparametric alternative, \citet{parast2019using} propose kernel-based procedures to test for a treatment effect using only surrogate marker information measured in a future study by borrowing information learned from a prior study. \citet{parast2023using} further develop this approach within a setting where there is heterogeneity in the utility of a surrogate. Heterogeneous surrogate utility occurs when a surrogate marker is more useful for certain subgroups of patients than others \citep{roberts2021, roberts2024surrogacy, parast2023testing, knowlton2024}. When this is true, one must account for heterogeneity when using the surrogate to test for a treatment effect in a future study \citep{parast2023using}. 

Notably, while \citet{parast2023using} accounts for heterogeneity, the approach still requires that the surrogate is sufficiently strong, in the sense of capturing the treatment effect, for all patients. In contrast, when the surrogate is strong for some patients and weak for others, the approach of \citet{parast2023using} is not appropriate, and it is not clear how the surrogate can be used appropriately to replace the primary outcome in a future study. In this paper, we are specifically focused on this setting where the surrogate is strong for only a subgroup of the study population. We propose a fully nonparametric kernel-based testing procedure for efficient testing using surrogate information (ETSI), where we replace the primary outcome with the surrogate information for only an appropriately selected subset of a heterogeneous population, and the primary outcome is purposefully measured for those where the surrogate would be considered weak. The ability to combine this information represents a significant advancement in the literature, as cost savings are still possible in cases where the surrogate is not sufficiently strong to use for the entire population, thus making methods for assessing heterogeneous utility in surrogate markers useful for future clinical trial design. In Section \ref{setting}, we describe our setting and assumptions. Section \ref{testingprocedure} outlines the proposed testing procedure and Section \ref{design} discusses future study design. We evaluate the performance of our method and compare to existing methods via a simulation study in Section \ref{sims} and illustrate our method via application to two distinct HIV clinical trials in Section \ref{example}. Lastly we include a discussion in Section \ref{discussion} of this contribution and possible extensions for future work.

\section{Setting and Assumptions }\label{setting}
\subsection{Notation and Setting} 

Let $G$ be the binary randomized treatment indicator, where $G=1$ indicates assignment to the treatment group and $G=0$ indicates assignment to the control group. Let $Y$ be the primary outcome, and $S$ be the surrogate marker where obtaining $Y$ is more expensive, burdensome, or requires longer follow-up compared to obtaining $S$; without loss of generality, we assume that higher values of $Y$ and $S$ are ``better." We assume that $S$ is continuous and measured after baseline, though $Y$ may be discrete or continuous. Additionally, let $\bW$ be a vector of baseline covariates of interest. Using a potential outcomes framework, each individual has possible outcomes $\{Y\supzero, Y\supone, S\supzero, S\supone\}$ depending on whether they are assigned to the treatment or control group, where $Y^{(g)}$ is the outcome when $G=g$ and $S^{(g)}$ is the surrogate when $G=g$. In practice the observed data are thus $\{Y_{ig},S_{ig},\bW_{ig}\}$ for each individual $i$. Let $n_g$ be the number of individuals in treatment group $g$.

We assume that for an identified region of the covariate space, $\Omega_\bW$, the surrogate has been evaluated to be sufficiently strong. Note that the following setup is agnostic to the method used to identify $\Omega_\bW$, and later in our numerical studies, we describe one possible construction of $\Omega_\bW$ based on the proportion of treatment effect explained. Let the prior clinical trial where the surrogate has been deemed sufficiently strong for patients in $\Omega_\bW$ be called Study A, where $Y$ and $S$ are fully measured. Let Study B denote the subsequent trial of interest, where our goal is to obtain $S$ as a replacement for $Y$ for individuals in $\Omega_{\bW}$ and to obtain $Y$ for others with a weak surrogate. 

Here, we now redefine the potential outcomes such that the study is explicit, i.e., the outcomes are $\{Y_K\supzero, Y_K\supone, S_K\supzero, S_K\supone\}$ for $K=A,B$. Let $\delta_{ig}$ be an indicator variable such that $\delta_{ig}=1$ when the surrogate is sufficiently strong, and thus we have measured $S_{ig}$ in Study B and omitted $Y_{ig}$ for individual $i$ in treatment group $g$ because $\bW_{ig} \in \Omega_\bW$. When $\delta_{ig}=0$, $S_{ig}$ is not measured and $Y_{ig}$ is purposefully measured because $\bW_{ig} \in \Omega_\bW^C$. We denote the observed data in the two studies with explicit superscripts, i.e., our data are i.i.d observations $\{Y_{ig}^A,S_{ig}^A,\bW_{ig}^A\}$ for each individual $i$ in Study A and i.i.d observations $\{Y_{jg}^B(1-\delta_{jg}^B),S_{jg}^B\delta_{jg},\bW_{jg}^B,\delta_{jg}^B\}$ for each individual $j$ in Study B.  Note that $Y_{jg}^B$ is only observed if $\delta_{jg}^B=0$ and $S_{jg}^B$ is only observed if $\delta_{jg}^B=1$. 

Figure \ref{fig_exampletables} illustrates our setting. Specifically, the top portion of Figure \ref{fig_exampletables} shows that in Study A, the surrogate $S_{ig}^A$ and the primary outcome $Y_{ig}^A$ are measured for the entire covariate space $W$. The middle portion shows that Study A data are used to \textit{identify} the region of strong surrogacy, $\Omega_W$, highlighted in pink. Then, in the bottom figure, in Study B, we purposefully \textit{only} measure $S_{jg}^B$ for individuals with $\bW_{jg}^B \in \Omega_W$ ($\delta_{jg}=1$ for these individuals), and \textit{only} measure $Y_{jg}^B$ for individuals with $\bW_{jg}^B \in \Omega_W^C$ ($\delta_{jg}=0$ for these individuals), where $-$ indicates not measured.

Our ultimate goal is to test the null hypothesis that there is no treatment effect on $Y$ in Study B:
\begin{eqnarray*}
H_0:&& \Delta_B \equiv E(Y_B\supone) - E(Y_B \supzero) =0.\\
H_1:&& \Delta_B \neq 0.
\end{eqnarray*}

\noindent Of course, one could carry out such a test by simply obtaining $Y$ for all individuals and performing, for example, a t-test. In contrast, our aim is to test $H_0$ by appropriately leveraging surrogate information, replacing $Y$ with $S$, such that we avoid measuring $Y$ for a specific subset of individuals. In the following section, we state our assumptions and then describe our proposed testing procedure.

\subsection{Assumptions \label{assumptions}}
We require the following conditions to hold within $\Omega_{\bW}$, the region of strong surrogacy:
\begin{enumerate}
\item[] (C1) $\nu_{K0}(s)$ and $\nu_{K1}(s)$ are monotone increasing in $s$, where $\nu_{Kg}(s) = E(Y_K\supg| \bW_K \in \Omega_{\bW},S_K\supg=s)$, for $K=A,B$;
\item[] (C2) $P(S_K\supone > s | \bW_K \in \Omega_{\bW}) \geq P(S_K\supzero > s| \bW_K \in \Omega_{\bW}) \forall s$, for $K=A,B$;
\item[] (C3) $\nu_{K1}(s) \geq \nu_{K0} (s) \forall s$, for $K=A,B$;
\item[] (C4) $\nu_{A0}(s) = \nu_{B0}(s) \forall s$;
\item[] (C5) Surrogacy heterogeneity with respect to $\bW$ in Study A and Study B is the same; and
\item [] (C6) $S_A\supone$, $S_A\supzero$, $S_B\supone$ and $S_B\supzero$ are continuous random variables with the same finite support over an interval $[a,b]$.
\end{enumerate}
In addition, within $\Omega_{\bW}^C$, we require:
\begin{enumerate}
\item[] (C7) $E(Y_K \supone) \geq E(Y_K \supzero)$, for $K=A,B$.
\end{enumerate}

\noindent Assumptions (C1)-(C3) are similar to those commonly required in the surrogate marker literature, as they are sufficient to guard against the surrogate paradox \citep{wang2002, taylor2005counterfactual, parast2016}. Specifically, (C1) guarantees that the surrogate and the primary outcome have a non-negative relationship. (C2) ensures a non-negative treatment effect on the surrogate marker within $\Omega_{\bW}$, while (C3) is equivalent to $\nu_1(s) - \nu_0(s) \geq 0$ for all $s$ and ensures that there is a non-negative residual treatment effect within $\Omega_{\bW}$, after accounting for the effect of the treatment on the surrogate. Additionally, we require the transportability assumptions (C4) and (C5), where (C4) requires the conditional mean functions in the control groups to be the same between Study A and Study B and (C5) requires the surrogacy heterogeneity to be the same. These assumptions are what allows us to use the information from Study A on individuals in Study B in $\Omega_\bW$. Clearly these are strong untestable assumptions, but we must assume something about the transportability between Study A and Study B; without it, it would not be realistic to think there is any validity in borrowing information from Study A. Assumption (C6) is required for implementing kernel smoothing in our nonparametric estimation procedure. Assumption (C7) parallels (C3) and ensures that there is a non-negative treatment effect in $\Omega_{\bW}^C$. While some of these assumptions may be explored empirically to some extent, they can be difficult to validate in practice. 

\section{Proposed Testing Procedure} \label{testingprocedure}

We consider the gold standard for testing $H_0$ to be simply obtaining $Y$ for all individuals in Study B. However, such an approach can be costly in many ways and fails to leverage surrogate information learned from Study A to alleviate financial costs or patient burden. At the other extreme, one could consider not obtaining $Y$ for anyone in Study B, and instead obtaining $S$ for everyone in Study B and using $S$ to replace $Y$ to test $H_0$. We refer to this existing approach as the surrogate-only testing procedure. Our proposed testing procedure will be a compromise between the gold standard and the surrogate-only testing procedure. To motivate our proposal, we describe the surrogate-only testing procedure below, and then our proposed testing procedure.

\subsection{Surrogate-only Testing Procedure}
 Prior work considers a surrogate-only testing procedure based on a treatment effect quantity that fully replaces the outcome with the surrogate in Study B, borrowing information from Study A \citep{parast2019using,parast2023testing}:
\begin{equation}\Delta_{AB} \equiv \int \mu_{A0}(s)dF_{B1}(s) - \int \mu_{A0}(s)dF_{B0}(s) \label{delta_ab}
\end{equation}
where $\mu_{K0}(s) = E(Y_K\supzero \mid S_K\supzero=s)$ and $F_{Bg}(s)$ is the cumulative distribution function of $S_{B} \supg$ (see Remark 1 for details about this construction). Note that the conditional mean $\mu_{A0}(s)$ is how information from Study A is used to infer the primary outcome in Study B. A testing procedure can then be constructed based on an estimate of $\Delta_{AB}$; we compare our proposed test to this approach in Section \ref{sims}. Importantly, this surrogate-only testing approach is inappropriate if the surrogate marker is an inadequate replacement of the outcome for some individuals. In the following section, we propose a testing procedure to handle such a setting.  

\vspace*{5mm}

\noindent \textit{Remark 1.} At first glance, it may be unclear why $\Delta_{AB}$ is defined such that $\mu_{A0}(s)$ is used in \textit{both} terms. Thus, we briefly explain the motivation behind this construction. First note that $\Delta_B$ itself can be expressed as 
\begin{eqnarray}
\Delta_B &=& E(Y_B\supone) - E(Y_B \supzero) \nonumber \\
&=& E_{S_B\supone}[E(Y_B\supone | S_B\supone=s)] - E_{S_B\supzero}[E(Y_B \supzero| S_B\supzero=s)] \nonumber\\
&=& \int \mu_{B1}(s)dF_{B1}(s) - \int \mu_{B0}(s)dF_{B0}(s). \label{delta_long}
\end{eqnarray}
Of course, the $\mu_{Bg}(s)$ components of (\ref{delta_long}) involve $Y$ from Study B. One could decide to simply replace these components with their parallel components from Study A, i.e., define $\Delta_{AB} $ as $\int \mu_{A1}(s)dF_{B1}(s) - \int \mu_{A0}(s)dF_{B0}(s)$. Now, we could estimate this without needing any information about $Y$ from Study B. However, one problem with this approach is that it can be shown that unless $S$ is a \textit{perfect} surrogate marker, it would be possible for this naive $\Delta_{AB} > \Delta_B$, and in an extreme case where $S$ is a poor surrogate, this naive $\Delta_{AB}$ may incorrectly indicate a treatment effect when in fact, $\Delta_B=0$. Therefore, \citet{parast2023testing} argues to replace \textit{both} conditional means in (\ref{delta_long}) with $\mu_{A0}(s)$, borrowing only from the control group of Study A, which will guarantee that $\Delta_{AB}\leq \Delta_B$. An estimate of this quantity would not require any information about $Y$ from Study B and thus would save the cost of measuring the primary outcome. It will, by design, be a conservative estimate of $\Delta_B$, but we view this as an advantage because it effectively provides a bound. 

\subsection{Proposed Pooled Treatment Effect Quantity} \label{pooled_quantity}

In contrast to $\Delta_{AB}$, we now propose a treatment effect quantity that only replaces the primary outcome with the surrogate marker in the subgroup where the surrogate is strong, i.e., $\Omega_{\bW}$, and uses the primary outcome for  $\Omega_{\bW}^C$. Let $\pi_B = P(\bW_B \in \Omega_{\bW}) = P(\delta_B = 1)$, the probability that $\bW_B$ is contained within $\Omega_{\bW}$, which is the same in each treatment group when treatment is randomized. Note that: 
\begin{eqnarray*}
     E(Y_B\supone) &=& E(Y_B\supone| \bW_B \in \Omega_{\bW})\pi_B + E(Y_B\supone| \bW_B \in \Omega_{\bW}^C)(1-\pi_B) \\
     &=& E_{S_B\supone} [E(Y_B\supone| \bW_B \in \Omega_{\bW},S_B\supone=s)]\pi_B + E(Y_B\supone| \bW_B \in \Omega_{\bW}^C)(1-\pi_B) \\
     &=& \pi_B \int \nu_{B1}(s)dF_{B1|\Omega_{\bW}}(s) + (1-\pi_B) \int y dG_{B1|\Omega_{\bW}^C}(y),  
 \end{eqnarray*}  where $\nu_{Bg}(s) = E(Y_B\supg| \bW_B \in \Omega_{\bW},S_B\supg=s)$, $F_{B1|\Omega_{\bW}}(s)$ is the cumulative distribution function of $S_{B} \supone$ given $\bW_B \in \Omega_{\bW}$, and $G_{B1|\Omega_{\bW}^C}$ is the cumulative distribution function of $Y_{B} \supone$ given $\bW_B \in \Omega_{\bW}^C$. Thus, we define the following treatment effect quantity:
 \begin{eqnarray*}
     \Delta_P &\equiv& \pi_B \int \nu_{A0}(s)dF_{B1|\Omega_{\bW}}(s) + (1-\pi_B) \int y dG_{B1|\Omega_{\bW}^C}(y) \\&& \hspace*{1.5in} - \pi_B \int \nu_{A0}(s)dF_{B0|\Omega_{\bW}}(s) - (1-\pi_B) \int y dG_{B0|\Omega_{\bW}^C}(y)\\
     &=& \pi_B \left \{\int \nu_{A0}(s)dF_{B1|\Omega_{\bW}}(s) - \int \nu_{A0}(s)dF_{B0|\Omega_{\bW}}(s) \right \} \\&& \hspace*{1.5in} + (1-\pi_B) \left \{ \int y dG_{B1|\Omega_{\bW}^C}(y) - \int y dG_{B0|\Omega_{\bW}^C}(y) \right \},
 \end{eqnarray*}
 where, similar to the description in Remark 1, we replace the conditional expectations involving the surrogate with the conditional expectation from the control group in Study A, thus using Study A to infer the primary outcome in Study B when $\bW \in \Omega_{\bW}$. If $E(Y_A\supzero| \bW_A \in \Omega_{\bW},S_A\supzero=s) = E(Y_A\supzero| S_A\supzero=s)$, and $F_{Bg|\Omega_{\bW}}(s) = F_{Bg}(s)$ and $G_{Bg|\Omega_{\bW}^C}(y) = P(Y_B\supg \leq y),$ then $\Delta_P$ is equal to a straightforward pooled treatment effect quantity: $(1-\pi_B) \Delta_B + \pi_B \Delta_{AB}.$ However, we would generally not assume that these equalities to hold. Nonetheless, we refer to $\Delta_P$ as a ``pooled" quantity because it pools both $Y$ and $S$ information to quantify the treatment effect. 

 Under the assumptions detailed in Section \ref{assumptions}, it can be shown that the proposed $\Delta_P$ quantity has the following two desirable properties:
 
 \begin{subtheorem} \label{1a}
    If $\Delta_B = 0$, then $~\Delta_P = 0$. 
\end{subtheorem}
\begin{subtheorem} \label{1b}
    $\Delta_P \leq \Delta_B$.
\end{subtheorem}
\noindent Proofs are provided in Appendix A. The first property (Theorem \ref{1a}) means that when there is truly no treatment effect on $Y$ in Study B, then there will also be no treatment effect on the pooled quantity. The second property (Theorem \ref{1b}) ensures that the pooled treatment effect provides a lower bound on the true treatment effect on $Y$. These properties are key to the validity of our testing procedure described in Section \ref{proposed} which is based on this pooled treatment effect. 

\vspace*{5mm}
\noindent \textit{Remark 2.} In practice, $\nu_{A0}(s)$ is not a random quantity because Study A data are fixed and known (i.e., $n_A$  does not $\rightarrow \infty$) in this testing framework. That is, we are specifically focused on a setting where Study A has concluded and we are making inference on Study B. Thus, an important statistical assumption we make is that the Study A data is conditioned on, and treated as fixed quantities in the probabilistic calculations. This means that our treatment effect quantities that involve Study A must be defined in a way that acknowledges and makes explicit this reliance on Study A. To this end, we define:  
 \begin{eqnarray*}
\Delta_{P|A}  &\equiv& \pi_B \left \{\int \widehat{\nu}_{A0}(s)dF_{B1|\Omega_{\bW}}(s) - \int \widehat{\nu}_{A0}(s)dF_{B0|\Omega_{\bW}}(s) \right \} \\ && \hspace*{1.5in} + (1-\pi_B) \left \{ \int y dG_{B1|\Omega_{\bW}^C}(y) - \int y dG_{B0|\Omega_{\bW}^C}(y) \right \}
\end{eqnarray*}
where
\begin{equation}
\widehat{\nu}_{A0}(s) = \frac{\sum_{i:\bW_i^A \in \Omega_{\bW}} K_{h}(S_{A0i} - s)Y_{A0i}}{\sum_{i:\bW_i^A \in \Omega_{\bW}} K_{h}(S_{A0i} - s)} \label{nu0_cond}
\end{equation}
is a consistent estimate of $\nu_{A0}(s)$.  Here, $K_h(\cdot) = K(\cdot/h)/h$ where $K(\cdot)$ is a smooth symmetric density function with finite support (e.g., standard normal density) and $h$ is a specified bandwidth, which may be data dependent. Note that this is equivalent to the Nadaraya-Watson conditional mean estimate \citep{nadaraya1964estimating, watson1964smooth}. 

 \subsection{Proposed Testing Procedure} \label{proposed}
We now propose a testing procedure based on a nonparametric estimate of $\Delta_{P|A}$.  Specifically, to test the null hypothesis that $\Delta_B=0$, we propose to instead test 
\begin{eqnarray*}
H_{0P}:&& \Delta_{P|A} = 0\\
H_{1P}:&& \Delta_{P|A} \neq 0.
\end{eqnarray*}
First, we construct a nonparametric estimate of $\Delta_{P|A}$. For each individual $i$ with $\bW_{ig} \in \Omega_{\bW}$, define $\tilde Y_{ig}^B = \widehat \nu_{A0} (S_{ig}^B),$ where $\widehat \nu_{A0}(s)$ is defined in (\ref{nu0_cond}) 
and let $$\widehat{\Delta}_{P|A}  = n_{B1}^{-1} \sum_{i=1}^{n_{B1}}\left[\delta_{i1} \tilde Y_{i1}^B + (1-\delta_{i1})  Y_{i1}^B\right] - n_{B0}^{-1} \sum_{i=1}^{n_{B0}}\left[\delta_{i0} \tilde Y_{i0}^B + (1-\delta_{i0}) Y_{i0}^B\right].$$

\noindent Note that this proposed estimator uses $Y$ for those with $\delta_{ig}=0$ and uses $S$ for those with $\delta_{ig}=1$, where $S$ is used to approximate $Y$ based on the learned conditional mean function from Study A without making any distributional assumptions.  In Appendix B, we show that $\widehat{\Delta}_{P|A}$ is a consistent estimate of $\Delta_{P|A}$ and that $\sqrt{n_B}(\widehat{\Delta}_{P|A}-\Delta_{P|A})$ converges to a mean zero normal distribution, with variance $\sigma^2_{P|A}$, which may be estimated by $\widehat{\sigma}^2_{P|A}$ with a closed form provided in Appendix B. Thus, we estimate the variance of $\widehat{\Delta}_{P|A}$ as:
\begin{eqnarray*}
  n_B^{-1}\widehat{\sigma}^2_{P|A} &=& n_{B1}^{-1} \left \{ (1-\widehat{\pi}_{B1}) s_1^2 + \widehat{\pi}_{B1} s_2^2  + \widehat \pi_{B1} (1-\widehat \pi_{B1})(\bar y_{B1C} - \bar y_{B1W})^2\right \} + \\ && \hspace{0.5in} n_{B0}^{-1}  \left \{ (1-\widehat{\pi}_{B0})s_3^2 +  \widehat{\pi}_{B0} s_4^2 + \widehat \pi_{B0} (1-\widehat \pi_{B0}) (\bar y_{B0C} - \bar y_{B0W})^2 \right \},
\end{eqnarray*}
where $s_1^2,s_2^2,s_3^2,s_4^2$ are empirical variances such that $s_1^2 = \var_{i \in \Omega_{\bW}^C}\left \{ Y_{i1}^B \right \}$, $s_2^2 = \var_{i \in \Omega_{\bW}} \left \{\tilde Y_{i1}^B \right\}$, $s_3^2 = \var_{i \in \Omega_{\bW}^C}\left \{  Y_{i0}^B \right \},$ and $s_4^2 =\var_{i \in \Omega_{\bW}} \left \{\tilde Y_{i0}^B \right \}$; $n_{BgC}$ and $n_{BgW}$ are the number of individuals in treatment group $g$ with $i\in \Omega_{\bW}^C$ and $i\in\Omega_{\bW}$, respectively; $\widehat{\pi}_{Bg} = n_{BgW}/n_{Bg}$, and $\bar y_{BgC}$ and $\bar y_{BgW}$ are the empirical means in treatment group $g$ with $i\in \Omega_{\bW}^C$ and $i\in\Omega_{\bW}$, respectively. Finally, we construct a Wald-type test statistic: 
$$\calt = \frac{\sqrt{n}_B\widehat \Delta_{P|A}}{\widehat \sigma_{P|A}},$$
and reject $H_{0P}$ when $|\calt| > \Phi^{-1}(1-\alpha/2)$, where $\Phi$ denotes the cumulative distribution function of a standard normal, and $\alpha$ is the desired level of the test. We examine the performance of this test in Section \ref{sims}. 

Notice that if Study B is conducted as we have specified, we cannot estimate $\Delta_B$ or $\Delta_{AB}$ because we have neither $Y$ nor $S$ for \textit{all} individuals in Study B; an estimate of $\Delta_B$ would require $Y$ for all individuals in Study B and an estimate of $\Delta_{AB}$ would require $S$ for all individuals in Study B. However, because we wish to compare our proposed estimate to these hypothetical estimates in the simulation study, we define those estimates here. For $\Delta_B$, if one did have (only) $Y$ for all individuals, estimation is straightforward:
$$\widehat \Delta_B = n_{B1}^{-1} \sum_{i=1}^{n_{B1}} Y_{i1}^B - 
 n_{B0}^{-1}\sum_{i=1}^{n_{B0}} Y_{i0}^B.$$
For $\Delta_{AB}$, if one did have (only) $S$ for all individuals, the resulting estimate would be:

$$\widehat \Delta_{AB} = n_{B1}^{-1} \sum_{i=1}^{n_{B1}} \widehat \mu_{A0}(S_{i1}^B) - 
 n_{B0}^{-1}\sum_{i=1}^{n_{B0}} \widehat \mu_{A0}(S_{i0}^B),$$ where 
 \begin{equation}
\widehat{\mu}_{A0}(s) = \frac{\sum_{i=1}^{n_{A0}} K_{h}(S_{A0i} - s)Y_{A0i}}{\sum_{i=1}^{n_{A0}} K_{h}(S_{A0i} - s)}, \label{mu0}
\end{equation}
the nonparametric estimate of the conditional mean, $\mu_{A0}(s)$.
Of course, this estimator is inappropriately using $S$ to predict $Y$ for all individuals, even those for whom the surrogate is weak. Testing procedures based on the Wald-type test statistics using $\widehat{\Delta}_B$ and $\widehat \Delta_{AB}$ can be similarly constructed as described above; we compare these to our proposed approach in Section \ref{sims}.  

\section{Study Design} \label{design}
In this section, we focus on a setting in which one plans to use our proposed testing procedure and is in the process of \textit{designing} Study B. That is, Study A has been completed and has been used to identify $\Omega_{\bW}$ (the subset where the surrogate is strong), and the goal is to optimally design Study B leveraging Study A information via the use of purposeful surrogate marker measurement. The null and alternative hypotheses of interest are:
\begin{eqnarray*}
    H_0: && \Delta_B=0\\
    H_1: && \Delta_B = \Psi,
\end{eqnarray*}
where we assume $\Psi>0$ without loss of generality. The power of our proposed test at level $\alpha=0.05$ is:
\begin{eqnarray*} \calp_B(\Psi) &=&  P\left ( \left |\frac{(1-\widehat{\pi}_{B})\widehat{\Delta}_{B|\Omega_{\bW}^C}+ \widehat{\pi}_{B} \widehat{\Delta}_{B|\Omega_{\bW}}}{\widehat \sigma_{P|A}/\sqrt{n_B}} \right| > 1.96 \mid \Delta_B = \Psi \right )\\
&& \mbox{with } \widehat{\Delta}_{B|\Omega_{\bW}^C} = [n_{B1C}^{-1} \sum_{i=1}^{n_{B1}}(1-\delta_{i1}) 
Y_{i1}^B - n_{B0C}^{-1} \sum_{i=1}^{n_{B0}}(1-\delta_{i0}) Y_{i0}^B]\\
&& \mbox{and } \widehat{\Delta}_{B|\Omega_{\bW}} = [ n_{B1W}^{-1} \sum_{i=1}^{n_{B1}}\delta_{i1} \tilde Y_{i1}^B  - n_{B0W}^{-1} \sum_{i=1}^{n_{B0}} \delta_{i0} \tilde Y_{i0}^B],
\end{eqnarray*}
where $\widehat \pi_{B} = (n_{B0W}+n_{B1W})/n_B$, the estimated $\pi$ in Study B. Notably, the expression above involves the specified alternative, $\Delta_B = \Psi$, which is generally the quantity that is provided when one asks for a power calculation. However, to calculate power, we will require the alternative specified not as $\Delta_B = \Psi$, but instead as the components $\Delta_{B|\Omega_{\bW}}$ and $\Delta_{B|\Omega_{\bW}^C}$, which are the treatment effects within $\Omega_W$ and $\Omega_W^C$, respectively, where 
\begin{eqnarray*}
\Delta_{B|\Omega_{\bW}} &=& \int \widehat{\nu}_{A0}(s)dF_{B1|\Omega_{\bW}}(s) - \int \widehat{\nu}_{A0}(s)dF_{B0|\Omega_{\bW}}(s) \\
\Delta_{B|\Omega_{\bW}^C} &=&\int y dG_{B1|\Omega_{\bW}^C}(y) - \int y dG_{B0|\Omega_{\bW}^C}(y)
\end{eqnarray*}
\noindent Since it would likely be unreasonable to expect a user to specify these components, we assume the specified alternative is given as $\Delta_B = \Psi$ and we define 
\begin{eqnarray*}
    \tau_K &=& \Delta_{K|\Omega_{\bW}^C}/ \Delta_K, \text{ and } \\
    \rho_K &=& \Delta_{K|\Omega_{\bW}}/ \Delta_K
\end{eqnarray*}
which we will use to translate the specified alternative into the components needed for power estimation.

Since the goal in this section is to \textit{plan} the future Study B, we must assume that only Study A data is available for the power calculation. To use Study A to guide the study design we assume, along with (C1)-(C7): 

\begin{enumerate}
\item[] (C8) The components of $\sigma_{A|P}$ and the quantities $\rho_K$ and $\tau$ are transportable from Study A to Study B e.g., $ \rho_A=\rho_B$ and $\tau_A=\tau_B$.
\end{enumerate}

\noindent This assumption is only needed for the study design procedure in this section. The quantities in (C8) that are assumed to be transportable can be estimated using generalized cross-validation (GCV) with Study A data only; in our numerical studies, we use GCV with 100 iterations with a holdout rate of 0.5 \citep{golub1979generalized}. Thus, the expected power of Study B, allowing $n_{B1},n_{B0}, \Psi,$ and $\pi_B$ to be user-specified is
\begin{eqnarray*} \calp_B(\Psi) &=& 1- \Phi \left (1.96 -  \frac{  (1-\pi_{B}) \widetilde{\tau}_A \Psi + \pi_{B}\widetilde{\rho}_A\Psi }{(n_{B0} + n_{B1})^{-1/2}\widetilde{\sigma}_{P|A}} \right),
\end{eqnarray*}
where ~$\widetilde{}$~ denotes a quantity estimated using GCV in Study A, $\widetilde{\rho}_A = \widetilde{\Delta}_{A|\Omega_{\bW}^C}/\widetilde{\Delta}_A,$  $\widetilde{\tau}_A = \widetilde{\Delta}_{A|\Omega_{\bW}}/\widetilde{\Delta}_A$, and 
\begin{eqnarray*}
(n_{B0} + n_{B1})^{-1}\widetilde{\sigma}_{P|A}^2 &=& n_{B1}^{-1} \left \{ (1-\pi_{B}) \widetilde{s}_1^2 + \pi_{B} \widetilde s_2^2  + \pi_{B} (1- \pi_{B})(\widetilde y_{A1C} - \widetilde y_{A1W})^2\right \} + \\ && \hspace{0.5in} n_{B0}^{-1}  \left \{ (1-\pi_{B}) \widetilde s_3^2 +  \pi_{B} \widetilde s_4^2 + \pi_{B} (1-\pi_{B}) (\widetilde y_{A0C} - \widetilde y_{A0W})^2 \right \}.
\end{eqnarray*}

\noindent We can also rearrange this expression to solve for the required sample size to achieve a desired power, given $\Psi$ and $\pi_B$. Suppose the desired power is $1-\beta$, i.e., $\calp_B(\Psi) = 1-\beta;$ then
\begin{eqnarray*}
n &=& \left \{\frac{1.96 - \Phi^{-1}(\beta)}{(1-\pi_{B}) \widetilde{\tau}_A \Psi + \pi_{B} \widetilde{\rho}_A\Psi } \right \}^2 \{ (1-\pi_{B}) \widetilde{s}_1^2 + \pi_{B} \widetilde s_2^2  + \pi_{B} (1- \pi_{B})(\widetilde y_{A1C} - \widetilde y_{A1W})^2 \\ &&\hspace*{1.5cm} +  (1-\pi_{B}) \widetilde s_3^2 +  \pi_{B} \widetilde s_4^2 + \pi_{B} (1-\pi_{B}) (\widetilde y_{A0C} - \widetilde y_{A0W})^2  \}
\end{eqnarray*}
\noindent where for simplicity, $n=n_{B1} = n_{B0}$. In addition to the sample size in Study B, it is important to note that one could essentially tune $\pi_B$ to influence the expected power of Study B by changing the strictness of the requirement for strong surrogacy and/or by intentionally recruiting more or less individuals with $\bW_B \in \Omega_{\bW}$. By appropriately adjusting the strictness of the surrogate requirement (which we illustrate in our numerical studies) and the number of patients that will meet them (recruitment based on $\bW)$, one can theoretically adjust $\pi_B$ to achieve the desired power in Study B. In our numerical studies and HIV application, we explore how expected power varies across different sample sizes and thresholds for strong surrogacy, thus providing a framework to use a prior study to design a future study that achieves a desired level of power while reducing costs and/or patient burden.

\section{Simulation Study} \label{sims}
The goals of the simulation study were to demonstrate the performance of the proposed estimation and testing procedures 
under various settings featuring heterogeneous surrogate information. We examined the performance of the proposed pooled treatment effect quantity and the corresponding testing procedure compared to the gold standard and the surrogate-only testing procedure. Specifically, we considered the resulting point estimates, standard error estimates, the effect size for the test, and the empirical power. Additionally, we compared the empirical power to the estimated power using the study design procedures detailed in Section \ref{design}. 

In all settings, samples sizes are $(n_1^A, n_0^A) = (1000,1100)$ and $(n_1^B, n_0^B) = (500,400)$, there was a single baseline covariate $W$, and we constructed the region of strong surrogacy $\Omega_W$ as follows. First, we used the data from Study A to estimate the proportion of the treatment effect explained (PTE) with respect to $W$ as in \cite{parast2023testing}, denoted as $\widehat{R}_S(W)$ where values close to 1 indicate strong surrogacy and values close to 0 indicate weak surrogacy. Then, we defined $\Omega_W$ as the region of the covariate space where $\widehat{R}_S(W) > \kappa$ for some specified threshold $\kappa$. We explored performance across three different thresholds ($\kappa = 0.5, 0.6, 0.7$) which correspond to varying strictness of the requirement for strong surrogacy.

Data generation details are provided in Appendix C; here, we briefly describe each setting. Setting 1 featured an extreme case of heterogeneous surrogate utility,  where the surrogate was useless for half of the population ($R_S(W) = 0$) and strong for the other half ($R_S(W) = 0.79$). Setting 2 featured data where the proportion of the treatment effect explained may take on several possible values rather than two extremes. Specifically, the covariate space was split into four equally likely regions with varying levels of surrogate strength ($R_S(W) = 0, 0.25, 0.52, 0.83$). In Setting 3, there was no treatment effect; this setting was included to examine the Type 1 error of our proposed testing procedure. 

All settings used a standard normal density for the kernel $K(\cdot)$ and were summarized over 1000 iterations. To reflect the expected setting that Study A has already taken place and is considered fixed, Study A was fixed and the simulation iterations generated Study B data only. The bandwidth was calculated as $h = h_b (n_0^A)^{-0.2}$, where $h_b$ was obtained using the \texttt{bw.nrd} function in \texttt{R} \citep{scott1992multivariate}. R code to reproduce all simulation results are available at: \url{https://github.com/rebeccaknowlton/etsi-simulations}.

Table \ref{simres} shows the results of the simulation study in terms of the estimates of $\Delta_B$, $\Delta_{AB}$, and $\Delta_P$ for $\kappa = 0.5, 0.6, 0.7$ for Settings 1-3. Recall that the entire premise of this method is that our approach offers the ability to measure only $Y$ for some people, and only $S$ for others; therefore, in practice, we could not estimate $\Delta_B$ and $\Delta_{AB}$ using the observed data (because you would need $Y$ for all individuals to estimate $\Delta_B$ and $S$ for all individuals to estimate $\Delta_{AB}$) as we have done in this simulation study. We include them here only for comparison. Here, we see that $\Delta_P$ is between $\Delta_B$ and $\Delta_{AB}$, as expected, depending on the strictness of the criteria for strong surrogacy. Likewise, the power of the test to detect a treatment effect in Settings 1 and 2 using $\Delta_P$ is between the corresponding power of the tests for $\Delta_B$ and $\Delta_{AB}$. Throughout, the empirical standard errors are close to the average standard errors. Note that in Setting 1, the empirical power using $\Delta_P$ (0.830 when $\kappa = 0.7$) is quite close to the empirical power using the gold standard (0.882) while only measuring $Y$ in 50\% of the patients. Similarly, in Setting 2, the empirical power using $\Delta_P$ (0.916 when $\kappa = 0.5$) is close to the empirical power using the gold standard (0.957) while only measuring $Y$ in 50\% of the patients. This highlights the potential benefits of our approach, achieving power less than, but close to, the gold standard power while measuring $Y$ for only a subset of the patients.

Setting 3, the null setting with a true treatment effect of 0, shows appropriate Type 1 Error control with rates close to 0.05. Table \ref{powerres} shows the estimated versus empirical power for testing $H_0: \Delta_P = 0$ for $\kappa = 0.5, 0.6, 0.7$. The estimated power is calculated using only Study A from the simulation settings, and following the procedure proposed in Section \ref{design} using generalized cross validation. Throughout, we see the estimated power is close to the empirical power. These results demonstrate reasonable performance of our proposed methods in finite samples.

\section{Example} \label{example}

We illustrate the performance of our methods on real data from two randomized AIDS clinical trials. This study was reviewed and approved by the Institutional Review Board of the University of Texas at Austin. For both trials, we considered the outcome of interest to be plasma HIV-1 RNA at baseline minus plasma HIV-1 RNA at 24 weeks post-randomization. Note that a \textit{decrease} in HIV-1 RNA (viral load) over time represents clinical improvement; therefore, this definition ensures that positive values of our outcome indicate better patient health. Since RNA is historically considered expensive to measure \citep{calmy2007hiv}, the potential surrogate marker is CD4 cell count at 24 weeks minus CD4 cell count at baseline to 24 weeks, where an increase in CD4 indicates improvement. The baseline covariate of interest is baseline CD4 cell count, which is known to affect patient response to treatment and informs current clinical guidelines for AIDS treatment \citep{NIH_HIV}. 

For Study A, we use data from the AIDS Clinical Trials Group (ACTG) 320 study \citep{hammer1997controlled}, where $n_1^A=418$ and $n_0^A=412$. We obtain a nonparametric estimate of the PTE as a function of baseline CD4 cell count via kernel smoothing, depicted in Figure \ref{aids_plot}, and we consider two potential thresholds for strong surrogacy of $\widehat{R}_S(W) > \kappa$ where $\kappa = (0.5, 0.6)$. For Study B, we use data from the AIDS Clinical Trials Group (ACTG) 193A study \citep{henry1998randomized}, where $n_1^B = 28$ and $n_0^B=37$. We have investigated the performance of our simulation studies in small sample sizes similar to these studies, and while our methods can accommodate small sample sizes in Study B, we found that they perform more reliably when Study A has larger sample sizes such as in Section \ref{sims}, e.g., approximately 1000 subjects per treatment arm. This is particularly true for the study design portion, since we must split Study A in the generalized cross-validation procedure. Given that our Study A sample size is smaller in this example application, study design results should be interpreted with caution, though they can still provide useful insights into the estimated treatment effect and estimated power of future studies. 

Table \ref{aids_table} shows the results for the estimated $\Delta_P$ at two possible $\kappa$ thresholds, $\kappa=0.5$ and $0.6$, as well as $\Delta_B$ and $\Delta_{AB}$ for comparison. Previous work has suggested that change in CD4 cell count may not be a valid surrogate for change in plasma HIV-1 RNA for certain patients \citep{o1996changes,lin1993evaluating}, and this is especially clear for patients with baseline CD4 cell count greater than roughly 50 in Figure \ref{aids_plot}. Certainly, it wouldn't be appropriate to use $S$ in Study B for everyone, and thus one should interpret $\widehat{\Delta}_{AB}$ with caution, when it suggests the treatment effect is negligible ($\widehat{\Delta}_{AB}=0.006, p=0.928)$. Meanwhile, measuring $Y$ for everyone and estimating $\Delta_B$, while costly, provides evidence of a significant treatment effect in Study B ($\widehat{\Delta}_{B} = 0.395, p=0.019$).

Our proposed method allows us to combine the surrogate information for the subset of patients for whom it is appropriate, while still measuring $Y$ for the rest of the population. The pooled estimates offer some evidence of a nonzero treatment effect closer to the estimated $\Delta_B$, while still offering some cost savings. At the higher threshold for strong surrogacy ($\kappa =0.6$), we observe a trade-off, where the magnitude of the estimated treatment effect is larger than at the lower threshold ($\kappa=0.5$) and is closer to the gold standard estimator, and the p-value correspondingly decreases ($\widehat{\Delta}_P=0.274, p=0.111$ when $\kappa=0.5$ vs. $\widehat{\Delta}_P=0.294, p=0.095$ when $\kappa=0.6$). While the small sample size of Study B is a limitation for this example, our results illustrate the potential benefit of $\widehat{\Delta}_P$ being a compromise between the expensive $\widehat{\Delta}_B$ and the $\widehat{\Delta}_{AB}$ that was inappropriate for a heterogeneous population featuring many individuals for whom the surrogate should not be used to replace the outcome.

Additionally, using the ACTG 320 Study (Study A) as a starting point, we explored the design parameters for a hypothetical future trial (Study B) that plans to use our proposed testing procedure using the methods proposed in Section \ref{design}. Figure \ref{aids_design_plot} illustrates Study B's estimated power across three key variables: total sample size, hypothesized treatment effect ($\Psi$), and the strong surrogacy threshold ($\kappa$). The figure shows that, as expected, the statistical power increases with larger sample sizes and larger treatment effects, while also highlighting the effect of varying the $\kappa$ threshold. For example, using a $\kappa=0.5$ threshold and an alternative of $\Psi =0.75$, the estimated power for a total sample size of $n=100$ is 78\%, while the estimated power for a total sample size of $n=150$ is 91\%. Our proposed framework thus enables study designers to make informed decisions about the trade-off between cost savings and power in future trials where the surrogate marker has heterogeneous utility. 

\section{Discussion} \label{discussion}

We present a novel method for efficient testing using surrogate information (ETSI), designed to reduce costs in scenarios with heterogeneous surrogate utility. ETSI provides a framework to combine surrogate information for subpopulations with strong surrogacy and the primary outcome for those with weak surrogacy, thereby leveraging findings from previous studies to streamline future research and  improve cost effectiveness. We compared ETSI's pooled treatment effect and corresponding test against two extremes: fully measuring the outcome and fully substituting the surrogate. Results show that ETSI effectively balances these approaches, capitalizing on cost-saving opportunities without compromising accuracy when surrogacy strength varies across subpopulations. We provide guidelines for using ETSI to design future studies, including setting appropriate thresholds for strong surrogacy and achieving desired statistical power. The proposed method's performance is validated through numerical simulations and illustrated using AIDS clinical trial data from two randomized studies. An R package implementing our proposed methods, $\texttt{etsi}$, is available at \url{https://github.com/rebeccaknowlton/etsi}.

In our numerical studies examining study design, we have primarily focused on influencing $\pi_B = P(\bW_B \in \Omega_B)$ by setting a threshold for strong surrogacy. However, one can further influence $\pi_B$ and therefore the power of Study B by intentionally recruiting participants who meet the characteristics of strong surrogacy, that is, purposefully recruiting participants with $\bW \in \Omega_{\bW}$. Participant recruitment is another important element of study design, and while not fully explored here, the recruitment process is indeed an additional mechanism through which we can set the power of Study B. This idea of using strategic recruitment to increase the power of future studies is notably related to the subject of enrichment in clinical trials. Enrichment studies in clinical trials typically leverage prior information to selectively recruit patient subgroups more likely to respond to treatment, aiming to enhance efficiency, reduce costs, and increase statistical power \citep{temple2010enrichment, thall2021adaptive}. However, these studies often face challenges with generalizability. While our proposed method shares the goal of improving trial efficiency through intermediate analysis, it differs fundamentally from enrichment studies. ETSI includes individuals from both strong and weak surrogacy subgroups in Study B, measuring the primary outcome $Y$ for the latter. This approach maintains generalizability while achieving cost savings and preserving statistical power. While enrichment and surrogate-based methods have traditionally been distinct areas of research, recent work has begun to explore their intersection, as in \cite{wu2022incorporating}. Further investigation into combining surrogate markers with enrichment strategies could potentially yield even greater reductions in trial costs and improvements in efficiency. This area of study warrants additional research to fully understand its benefits and applications in clinical trial design.

While we have primarily focused on heterogeneous surrogacy as the motivation for ETSI, an interesting alternative setting emerges when surrogates are universally appropriate but some participants are unwilling or unable to have the primary outcome measured. However, ETSI in its current form is not directly applicable to this setting because the decision concerning which patients will have the surrogate vs. the primary outcome measured is set by design, not by patient choice; allowing patient choice introduces new complexities. Patients who are unwilling or unable to have the primary outcome measured may fundamentally differ from the rest of the patients in many ways, both measured and unmeasured, including with respect to their ultimate treatment response. For instance, patients declining primary outcome measurement may have increased frailty or sensitivity to the measurement burden of $Y$, and these differences could potentially introduce bias into the results. Addressing this scenario would require extending the ETSI framework to account for patient preferences and potential associated biases. Such an extension could support shared decision-making in patient care, aligning with modern healthcare trends \citep{dennison2023shared, muscat2021health}. %This situation bears similarities to data that are ``missing not at random'' (MNAR) in clinical settings, a well-studied problem outside of the surrogate context with many suitable methods summarized in \cite{little2019statistical}.  However, the implications of this issue in the surrogate setting, particularly when ``missingness'' is a result of patient choice, are not as well understood. %However, it would require careful consideration of potential selection biases, and additional research is needed to ensure the extension's validity and utility in clinical trials.

The proposed framework, while promising, has some notable limitations. First, our method relies on a set of assumptions that, while common in the surrogate marker literature, may be considered restrictive. Additionally, the effectiveness of our approach depends on having a sufficiently large sample size in Study A and assumes the existence of a reasonable approach to identify a surrogate as strong or weak according to patient characteristics in Study A. There is currently no consensus on the optimal method for evaluating surrogate strength, and this evaluation can be a complex problem. Lastly, a key assumption of both our testing and design approaches is that Study A is sufficiently representative of and generalizable to Study B such that information learned in Study A can be transferred to Study B. When this may not be true, one could consider extensions that borrow from recent research in transfer learning and domain adaptation \citep{kouw2019review,cai2024semi}; further research in this area is warranted.   

\vspace*{-5mm}
\section*{Acknowledgements}
This work was supported by NIDDK grant R01DK118354 (PI:Parast). We are grateful to the AIDS Clinical Trial Group (ACTG) Network for supplying the ACTG data.
\vspace*{-5mm}

\section*{Conflict of Interest Statement}
The authors have no conflicts of interest to disclose.
\vspace*{-5mm}

\section*{Data Availability Statement}
The data from the ACTG 320 study and the ACTG 193A study used in this paper are publicly available upon request from the AIDS Clinical Trial Group: \url{https://actgnetwork.org/submit-a-proposal}.
\vspace*{-5mm}

\bibliographystyle{biom}
 \bibliography{bib}

\clearpage\begin{figure}
    \centering
\captionsetup{font={small,stretch=1.1}} 

% Adjust cell height
\renewcommand{\arraystretch}{1.3}

\textbf{Study A:}
\[
\begin{array}{|c|c|c|}
\hline
 W_{1g}^A  &  S_{1g}^A & Y_{1g}^A \\ \hline
W_{2g}^A  & S_{2g}^A & Y_{2g}^A \\ \hline
 W_{3g}^A & S_{3g}^A & Y_{3g}^A \\ \hline
W_{4g}^A  & S_{4g}^A & Y_{4g}^A \\ \hline
\vdots & \vdots & \vdots \\ \hline
W_{ng}^A & S_{ng}^A & Y_{ng}^A \\
\hline
\end{array}
\]

\vspace{2em}  
\textbf{Study A, $\Omega_{\bW}$ identified:}
\[
\begin{array}{|c|c|c|}
\hline
\cellcolor{pink} W_{1g}^A  & \cellcolor{pink} S_{1g}^A & \cellcolor{pink} Y_{1g}^A \\ \hline
\cellcolor{pink} W_{2g}^A  & \cellcolor{pink} S_{2g}^A & \cellcolor{pink} Y_{2g}^A \\ \hline
\cellcolor{pink} W_{3g}^A & \cellcolor{pink} S_{3g}^A & \cellcolor{pink} Y_{3g}^A \\ \hline
W_{4g}^A  & S_{4g}^A & Y_{4g}^A \\ \hline
\vdots & \vdots & \vdots \\ \hline
W_{ng}^A & S_{ng}^A & Y_{ng}^A \\
\hline
\end{array}
\]

\vspace{2em}  

\textbf{Study B:}
\[
\begin{array}{|c|c|c|c|}
\hline
\cellcolor{pink} W_{1g}^B & \cellcolor{pink} \delta_{1g}^B = 1 &  \cellcolor{pink} S_{1g}^B & \cellcolor{pink} - \\ \hline
\cellcolor{pink} W_{2g}^B & \cellcolor{pink} \delta_{2g}^B = 1 & \cellcolor{pink} S_{2g}^B & \cellcolor{pink} - \\ \hline
\cellcolor{pink} W_{3g}^B & \cellcolor{pink} \delta_{3g}^B = 1 & \cellcolor{pink} S_{3g}^B & \cellcolor{pink} - \\ \hline
W_{4g}^B & \delta_{4g}^B = 0  & - & Y_{4g}^B \\ \hline
\vdots & \vdots & \vdots  & \vdots \\ \hline
W_{ng}^B & \delta_{ng}^B = 0  & - & Y_{ng}^B \\
\hline
\end{array}
\]

    \caption{ Illustration of our setting: Our observed data consists of $\{Y_{ig}^A,S_{ig}^A,\bW_{ig}^A\}$ for each individual $i$ in Study A and $\{Y_{jg}^B(1-\delta_{jg}^B),S_{jg}^B\delta_{jg},\bW_{jg}^B,\delta_{jg}^B\}$ for each individual $j$ in Study B.  Top figure: In Study A, the surrogate $S_{ig}^A$ and the primary outcome $Y_{ig}^A$ are measured for the entire covariate space $W$. Middle figure: Study A data are used to \textit{identify} the region of strong surrogacy, $\Omega_W$, highlighted in pink. Bottom figure: Then, in Study B, we purposefully \textit{only} measure $S_{jg}^B$ for individuals with $\bW_{jg}^B \in \Omega_W$ ($\delta_{jg}=1$ for these individuals), and \textit{only} measure $Y_{jg}^B$ for individuals with $\bW_{jg}^B \in \Omega_W^C$ ($\delta_{jg}=0$ for these individuals), where $-$ indicates not measured.}
    \label{fig_exampletables}
\end{figure} 

\clearpage
\begin{table}[hptb]
\caption{Estimation results for $\Delta_B$, $\Delta_{AB}$, and $\Delta_P$ under $\kappa = 0.5, 0.6, 0.7$ in Settings 1-3. ESE reflects the empirical standard error across all simulations, ASE reflects the average of the estimated standard error across all simulations, and Effect Size is the point estimate divided by the estimated standard error. Power and Type 1 Error represent the proportion of times the null hypothesis is rejected when there is truly a nonzero treatment effect and when the treatment effect is zero, respectively. Lastly, $\pi_B$ represents the true proportion of individuals who meet the requirement for strong surrogacy. (Note these are purposefully omitted for $\Delta_P$ in Setting 3, the null setting, because the PTE is undefined.)\label{simres}}

% Adjust cell height
\renewcommand{\arraystretch}{1}

\begin{center}
\begin{tabular}{|l|c|c|c|c|c|} \hline
% SETTING 1
&\multicolumn{5}{c|}{Setting 1}\\ \hline
\multicolumn{1}{|l|}{}&\multicolumn{1}{c|}{$\Delta_B$}&\multicolumn{1}{c|}{$\Delta_{AB}$}&\multicolumn{1}{c|}{$\Delta_P$: $\kappa = 0.5$}&\multicolumn{1}{c|}{$\Delta_P$: $\kappa = 0.6$}&\multicolumn{1}{c|}{$\Delta_P$: $\kappa = 0.7$}\\ \hline
Point Estimate&2.318&1.004&2.039&2.124&2.147\\ 
ESE&0.740&0.412&0.719&0.727&0.729\\ 
ASE&0.741&0.402&0.720&0.727&0.728\\ 
Effect Size&3.135&2.515&2.840&2.931&2.956\\ 
Power&0.882&0.699&0.815&0.835&0.838\\ \hline
$\pi_B$&0.000&1.000&0.500&0.500&0.500\\ 
\hline
% SETTING 2
&\multicolumn{5}{c|}{Setting 2}\\ \hline
\multicolumn{1}{|l|}{}&\multicolumn{1}{c|}{$\Delta_B$}&\multicolumn{1}{c|}{$\Delta_{AB}$}&\multicolumn{1}{c|}{$\Delta_P$: $\kappa = 0.5$}&\multicolumn{1}{c|}{$\Delta_P$: $\kappa = 0.6$}&\multicolumn{1}{c|}{$\Delta_P$: $\kappa = 0.7$}\\ \hline
Point Estimate&1.756&0.664&1.415&1.440&1.498\\ 
ESE&0.471&0.260&0.430&0.435&0.451\\ 
ASE&0.470&0.256&0.429&0.431&0.444\\ 
Effect Size&3.746&2.611&3.306&3.350&3.381\\ 
Power&0.957&0.735&0.916&0.928&0.922\\   \hline
$\pi_B$&0.000&1.000&0.500&0.250&0.250\\ 
\hline
% SETTING 3
&\multicolumn{5}{c|}{Setting 3}\\ \hline
\multicolumn{1}{|l|}{}&\multicolumn{1}{c|}{$\Delta_B$}&\multicolumn{1}{c|}{$\Delta_{AB}$}&\multicolumn{1}{c|}{$\Delta_P$: $\kappa = 0.5$}&\multicolumn{1}{c|}{$\Delta_P$: $\kappa = 0.6$}&\multicolumn{1}{c|}{$\Delta_P$: $\kappa = 0.7$}\\ \hline
Point Estimate&0.034&0.039&0.028&0.031&0.031\\ 
ESE&0.604&0.424&0.537&0.544&0.551\\ 
ASE&0.615&0.411&0.543&0.552&0.556\\ 
Effect Size&0.055&0.094&0.051&0.055&0.056\\ 
Type 1 Error&0.037&0.061&0.047&0.040&0.042\\  \hline 
$\pi_B$&0.000&1.000&   -&   -&   -\\  
\hline
\end{tabular}
\vspace{3mm}
\end{center}
\end{table}

\clearpage
\begin{table}[hptb]
\caption{Estimated versus empirical power for testing $H_0: \Delta_P = 0$ under $\kappa = 0.5, 0.6, 0.7$ in Settings 1-2 (Setting 3, the null setting, is omitted). $\pi_B$ represents the true proportion of individuals who meet the requirement for strong surrogacy. \label{powerres}}

\begin{center}
\begin{tabular}{|l|c|c|c|} \hline
% SETTING 1
&\multicolumn{3}{c|}{Setting 1}\\ \hline
\multicolumn{1}{|l|}{}&\multicolumn{1}{c|}{$\Delta_P$: $\kappa = 0.5$}&\multicolumn{1}{c|}{$\Delta_P$: $\kappa = 0.6$}&\multicolumn{1}{c|}{$\Delta_P$: $\kappa = 0.7$}\\ \hline
Estimated Power&0.784&0.860&0.854\\ 
Empirical Power&0.815&0.835&0.838\\ \hline $\pi_B$&0.500&0.500&0.500\\ 
\hline
% SETTING 2
&\multicolumn{3}{c|}{Setting 2}\\ \hline
\multicolumn{1}{|l|}{}&\multicolumn{1}{c|}{$\Delta_P$: $\kappa = 0.5$}&\multicolumn{1}{c|}{$\Delta_P$: $\kappa = 0.6$}&\multicolumn{1}{c|}{$\Delta_P$: $\kappa = 0.7$}\\ \hline
Estimated Power&0.949&0.940&0.918\\ 
Empirical Power&0.916&0.928&0.922\\  \hline
$\pi_B$ &0.500&0.250&0.250\\ 
\hline
\end{tabular}
\vspace{3mm}
\end{center}
\end{table}

\clearpage 
\begin{figure}
    \centering 
    \includegraphics[scale = 0.5]{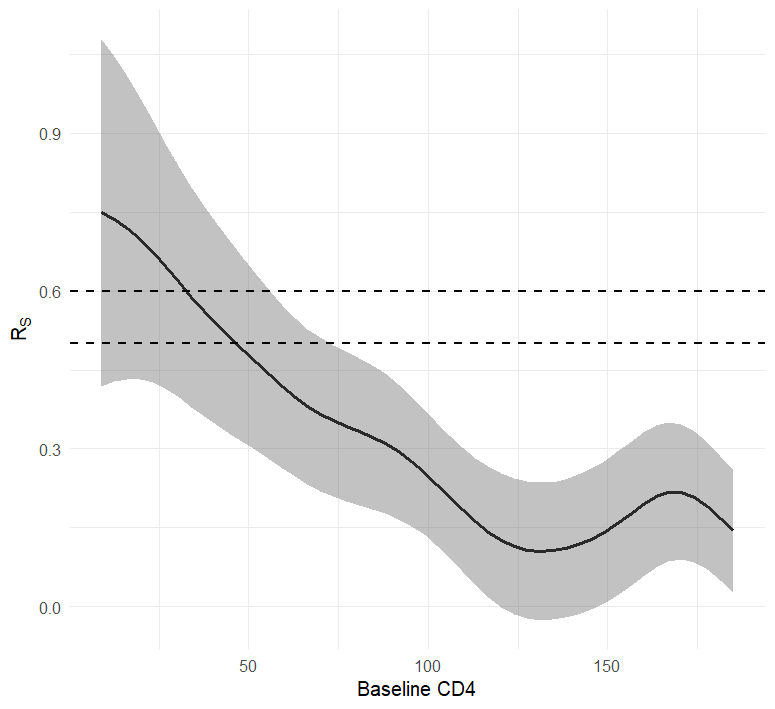}
    \caption{
    Estimated heterogeneity in $S$ (estimate of $R_S$) in ACTG 320 (Study A) plotted against the baseline CD4 cell count, the covariate of interest. The dashed lines indicate two potential thresholds for strong surrogacy: $\kappa = 0.5$ or $\kappa = 0.6$.
    } \label{aids_plot}
\end{figure}

\clearpage
\begin{table}[hptb]
\caption{Estimation results are shown for ACTG 320 as Study A and ACTG 193A as Study B. We provide the gold standard procedure $\Delta_B$ and the surrogate-only procedure $\Delta_{AB}$ for comparison, and compute the pooled estimator $\Delta_P$ at two different thresholds for strong surrogacy, $\kappa = 0.5$ and $\kappa =0.6$.\label{aids_table}}
\begin{center}
\begin{tabular}{|l|c|c|c|c|} \hline
\multicolumn{1}{|l|}{}&\multicolumn{1}{c|}{$\Delta_B$}&\multicolumn{1}{c|}{$\Delta_{AB}$}&\multicolumn{1}{c|}{$\Delta_P$: $\kappa = 0.5$}&\multicolumn{1}{c|}{$\Delta_P$: $\kappa = 0.6$}\\
\hline
Estimate&0.395&0.006&0.274&0.294\\ 
SE&0.168&0.062&0.172&0.176\\ 
Effect Size&2.353&0.090&1.593&1.672\\ 
p-value&0.019&0.928&0.111&0.095\\ 
\hline
\end{tabular}
\vspace{3mm}
\end{center}
\end{table}

\clearpage 
\begin{figure}
    \centering 
    \includegraphics[scale = 0.6]{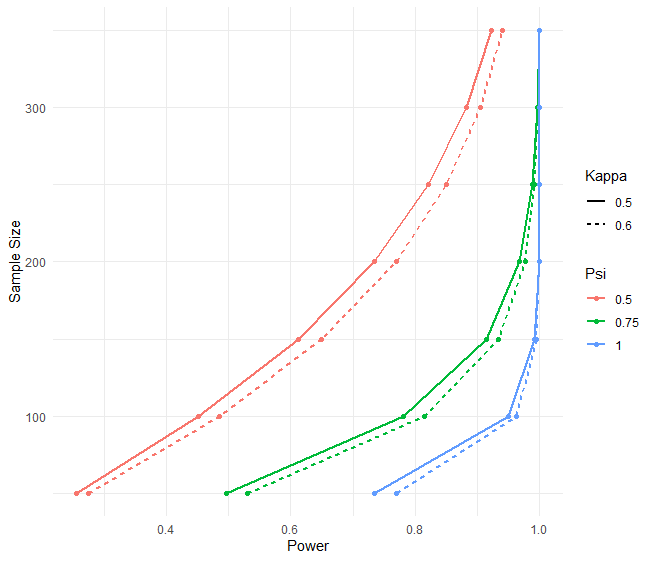}
    \caption{
    Estimated power for a future Study B given that Study A was ACTG 320, based on total sample size (displayed on the y-axis) and different options for the effect size to detect (Psi) and the threshold for strong surrogacy (Kappa).
    } \label{aids_design_plot}
\end{figure}

\clearpage
\appendix

\setcounter{table}{0}
\renewcommand{\thetable}{A\arabic{table}}
\setcounter{figure}{0}
\renewcommand{\thefigure}{A\arabic{figure}}
\renewcommand{\theequation}{A.\arabic{equation}}

\allowdisplaybreaks

\section*{Appendix A} 
 \setlength{\parindent}{0cm}
\textbf{Theorem 1.A} If $\Delta_B = 0$, then $~\Delta_P = 0$. 

 \begin{proof}
 We will show that if $\Delta_B = 0$, then $\Delta_P = 0$. Note that $\Delta_B$ can be expressed as 
\begin{eqnarray}
     \Delta_B &=& E(Y_B\supone) - E(Y_B\supzero )  \nonumber \\
     &=& \pi_B \int \nu_{B1}(s)dF_{B1|\Omega_{\bW}}(s) + (1-\pi_B) \int y dG_{B1|\Omega_{\bW}^C}(y) \nonumber\\ 
     && - \pi_B \int \nu_{B0}(s)dF_{B0|\Omega_{\bW}}(s) - (1-\pi_B) \int y dG_{B0|\Omega_{\bW}^C}(y) \nonumber\\
     &=& \pi_B \left \{ \int \nu_{B1}(s)dF_{B1|\Omega_{\bW}}(s) - \int \nu_{B0}(s)dF_{B0|\Omega_{\bW}}(s) \right \} + \label{eq1a}\\
     && (1-\pi_B) \left \{\int y dG_{B1|\Omega_{\bW}^C}(y) -  \int y dG_{B0|\Omega_{\bW}^C}(y) \right \}   \label{eq2a}
 \end{eqnarray} 
We will show that (\ref{eq1a}) and (\ref{eq2a}) are both $\geq 0$.  First, we state a needed lemma:
\begin{lemma}
If $g$ is a monotone increasing function, and $P(X > s) \geq P(Y > s) \, \forall s$, then 
\begin{equation}
    E\left[g(X)\right] \geq E\left[g(Y)\right]. \label{eqlemma}
\end{equation}
\end{lemma}

\begin{proof} 
Let $a$ and $b$ be the minimum and maximum, respectively, of the support of $X$ and $Y$. Then 
\begin{align*}
   E[g(X)] &= \int_a^b g(s)f_X(s)ds \\
    &= \left[g(s) F_X(s) \right]_a^b - \int_a^b g'(s) F_X(s)ds \\ 
    &= g(b) - \int_a^b g'(s)F_X(s)ds,
\end{align*}
where the integral above is obtained via integration by parts using $u = g(s)$ and $v = F_X(s)$ and $F_X$ denotes the CDF of $X$. Similarly, we can write 
\begin{align*}
    E[g(Y)] &= g(b) - \int_a^b g'(s) F_Y(s)ds.
\end{align*}

Because $P(X >s) \geq P(Y >s)\forall s$, then $F_X(s) \leq F_Y(s) \forall s$. Also, $g$ being monotone increasing implies that $g'(s) >0 \forall s$. Therefore,
\begin{align*}
    \int_a^b g'(s) F_X(s)ds &\leq \int_a^b g'(s) F_Y(s) ds \\
    -\int_a^b g'(s) F_X(s)ds &\geq -\int_a^b g'(s) F_Y(s) ds \\
    g(b) -\int_a^b g'(s) F_X(s)ds &\geq g(b)-\int_a^b g'(s) F_Y(s) ds \\ 
    E[g(X)] &\geq E[g(Y)], \text{ proving (\ref{eqlemma})}.
\end{align*}
\end{proof}

Now we examine (\ref{eq1a}):
\begin{eqnarray}
(\ref{eq1a}) &=& \pi_B \left \{ \int \nu_{B1}(s)dF_{B1|\Omega_{\bW}}(s) - \int \nu_{B0}(s)dF_{B0|\Omega_{\bW}}(s) \right \} \nonumber\\
&\geq& \pi_B \left \{ \int \nu_{B1}(s)dF_{B0|\Omega_{\bW}}(s) - \int \nu_{B0}(s)dF_{B0|\Omega_{\bW}}(s) \right \} \label{eq3a}\\
&=& \pi_B \left \{ \int [\nu_{B1}(s) - \nu_{B0}(s)] dF_{B0|\Omega_{\bW}}(s) \right \} \nonumber\\
&\geq& 0 \label{eq4a}
\end{eqnarray}
where (\ref{eq3a}) follows from Lemma 1, (C1), and (C2),  and (\ref{eq4a}) follows from (C3). 

Now we examine (\ref{eq2a}):
\begin{eqnarray}
(\ref{eq2a}) &=& (1-\pi_B) \left \{\int y dG_{B1|\Omega_{\bW}^C}(y) -  \int y dG_{B0|\Omega_{\bW}^C}(y) \right \} \nonumber \\
&=& (1-\pi_B) \left \{ E( Y_B\supone | \bW_B\supone \in \Omega_{\bW}^C) - E( Y_B\supzero | \bW_B\supzero \in \Omega_{\bW}^C) \right \} \nonumber \\
&\geq & 0 \label{eq5a}
\end{eqnarray}
where (\ref{eq5a}) follows from (C7). It then follows that when $\Delta_B = 0$, it must be the case that (\ref{eq1a}) $=0$ and (\ref{eq2a}) $=0$. Now we examine $\Delta_P$ which is defined as:
 \begin{eqnarray}
     \Delta_P &=& \pi_B \left \{\int \nu_{A0}(s)dF_{B1|\Omega_{\bW}}(s) - \int \nu_{A0}(s)dF_{B0|\Omega_{\bW}}(s) \right \} \label{eq6a} \\&&  + (1-\pi_B) \left \{ \int y dG_{B1|\Omega_{\bW}^C}(y) - \int y dG_{B0|\Omega_{\bW}^C}(y) \right \}. \label{eq7a}
 \end{eqnarray}
Note that (\ref{eq7a}) is equal to (\ref{eq2a}). Thus, when $\Delta_B=0$, this component is also 0. It remains to show that (\ref{eq6a}) $=0$ when $\Delta_B=0$; we will do this by showing that when $\Delta_B=0$, it follows that $(\ref{eq6a})\geq0$ and $(\ref{eq6a}) \leq 0$ and thus $(\ref{eq6a})=0$ (i.e., bound from below and above). First,

\begin{eqnarray}
    (\ref{eq6a}) &=& \pi_B \left \{\int \nu_{A0}(s)dF_{B1|\Omega_{\bW}}(s) - \int \nu_{A0}(s)dF_{B0|\Omega_{\bW}}(s) \right \} \nonumber \\
    &=& \pi_B \left \{\int \nu_{B0}(s)dF_{B1|\Omega_{\bW}}(s) - \int \nu_{B0}(s)dF_{B0|\Omega_{\bW}}(s) \right \} \label{eq8a}\\
    &=& \pi_B \left \{\int \nu_{B0}(s)dF_{B1|\Omega_{\bW}}(s) - \int \nu_{B1}(s)dF_{B1|\Omega_{\bW}}(s) \right \} \label{eq9a} \\
    &=& \pi_B \left \{\int [ \nu_{B0}(s) - \nu_{B1}(s)] dF_{B1|\Omega_{\bW}}(s) \right \} \nonumber \\
    \leq 0 \label{eq10a}
\end{eqnarray}
where (\ref{eq8a}) follows from (C4), (\ref{eq9a}) follows from $(\ref{eq1a})=0$, and (\ref{eq10a}) follows from (C3). And then, 

\begin{eqnarray}
    (\ref{eq6a}) &=& \pi_B \left \{\int \nu_{A0}(s)dF_{B1|\Omega_{\bW}}(s) - \int \nu_{A0}(s)dF_{B0|\Omega_{\bW}}(s) \right \} \nonumber \\
        \geq 0 
\end{eqnarray}
which follows directly from Lemma 1, (C1), and (C2).  Thus, we have shown that if $\Delta_B=0$, then $\Delta_P=0$. 
\end{proof}
\vspace*{1cm}

\textbf{Theorem 1.B}  $\Delta_P \leq \Delta_B$.
\begin{proof}
\begin{eqnarray}
     \Delta_P &=& \pi_B \left \{\int \nu_{A0}(s)dF_{B1|\Omega_{\bW}}(s) - \int \nu_{A0}(s)dF_{B0|\Omega_{\bW}}(s) \right \} \nonumber \\&&  + (1-\pi_B) \left \{ \int y dG_{B1|\Omega_{\bW}^C}(y) - \int y dG_{B0|\Omega_{\bW}^C}(y) \right \} \nonumber\\
     &=& \pi_B \left \{\int \nu_{B0}(s)dF_{B1|\Omega_{\bW}}(s) - \int \nu_{B0}(s)dF_{B0|\Omega_{\bW}}(s) \right \} \label{eq11a} \\&&  + (1-\pi_B) \left \{ \int y dG_{B1|\Omega_{\bW}^C}(y) - \int y dG_{B0|\Omega_{\bW}^C}(y) \right \}, \nonumber\\
          &\leq& \pi_B \left \{\int \nu_{B1}(s)dF_{B1|\Omega_{\bW}}(s) - \int \nu_{B0}(s)dF_{B0|\Omega_{\bW}}(s) \right \} \label{eq12a} \\&&  + (1-\pi_B) \left \{ \int y dG_{B1|\Omega_{\bW}^C}(y) - \int y dG_{B0|\Omega_{\bW}^C}(y) \right \}, \nonumber \\
          &=& \Delta_B \nonumber
 \end{eqnarray}
 where (\ref{eq11a}) follows from (C4) and (\ref{eq12a}) follows from (C3). 
\end{proof}

\section*{Appendix B}
Here, we show that $\widehat{\Delta}_{P|A}$ is a consistent estimate of $\Delta_{P|A}$ and that $\sqrt{n_B}(\widehat{\Delta}_{P|A}-\Delta_{P|A})$ converges to a mean zero normal distribution. We assume $n_{B1}/n_B = p_1>0$ and $n_{B0}/n_B = p_0>0$.

First, we have 
\begin{eqnarray*}
    \widehat{\Delta}_{P|A}  &=& n_{B1}^{-1} \sum_{i=1}^{n_{B1}}\left[(1-\delta_{i1}) Y_{i1}^B + \delta_{i1} \tilde Y_{i1}^B\right] - n_{B0}^{-1} \sum_{i=1}^{n_{B0}}\left[(1-\delta_{i0}) Y_{i0}^B + \delta_{i0} \tilde Y_{i0}^B\right] \\
    &=& n_{B1}^{-1} \sum_{i=1}^{n_{B1}}(1-\delta_{i1}) Y_{i1}^B + n_{B1}^{-1} \sum_{i=1}^{n_{B1}}\delta_{i1}\tilde Y_{i1}^B - n_{B0}^{-1} \sum_{i=1}^{n_{B0}}(1-\delta_{i0})Y_{i0}^B - n_{B0}^{-1} \sum_{i=1}^{n_{B0}}\delta_{i0}\tilde Y_{i0}^B\\
    &=& n_{B1}^{-1} \sum_{i=1}^{n_{B1C}} Y_{i1}^B + n_{B1}^{-1} \sum_{i=1}^{n_{B1W}}\tilde Y_{i1}^B - n_{B0}^{-1} \sum_{i=1}^{n_{B0C}}Y_{i0}^B - n_{B0}^{-1} \sum_{i=1}^{n_{B0W}}\tilde Y_{i0}^B \\ 
    &=& \frac{n_{B1C}}{n_{B1}} \sum_{i=1}^{n_{B1C}} \frac{Y_{i1}^B}{n_{B1C}} + \frac{n_{B1W}}{n_{B1}} \sum_{i=1}^{n_{B1W}}\frac{\tilde Y_{i1}^B}{n_{B1W}} - \frac{n_{B0C}}{n_{B0}}  \sum_{i=1}^{n_{B0C}} \frac{Y_{i0}^B}{n_{B0C}} - \frac{n_{B0W}}{n_{B0}} \sum_{i=1}^{n_{B0W}} \frac{\tilde Y_{i0}^B}{n_{B0W}} \\
    &=& (1- \widehat \pi_{B1}) \sum_{i=1}^{n_{B1C}} \frac{Y_{i1}^B}{n_{B1C}} + \widehat \pi_{B1} \sum_{i=1}^{n_{B1W}}\frac{\tilde Y_{i1}^B}{n_{B1W}} - (1- \widehat  \pi_{B0}) \sum_{i=1}^{n_{B0C}} \frac{Y_{i0}^B}{n_{B0C}} - \widehat \pi_{B0} \sum_{i=1}^{n_{B0W}} \frac{\tilde Y_{i0}^B}{n_{B0W}}.
\end{eqnarray*}

% need to weave this in better
Let $\widehat \pi_{B0} \approx \widehat \pi_{B1} = \widehat \pi_B$, and thus,
\begin{eqnarray*}
    \widehat{\Delta}_{P|A}  &=& (1- \widehat \pi_{B}) \sum_{i=1}^{n_{B1C}} \frac{Y_{i1}^B}{n_{B1C}} + \widehat \pi_{B} \sum_{i=1}^{n_{B1W}}\frac{\tilde Y_{i1}^B}{n_{B1W}} - (1- \widehat  \pi_{B}) \sum_{i=1}^{n_{B0C}} \frac{Y_{i0}^B}{n_{B0C}} - \widehat \pi_{B} \sum_{i=1}^{n_{B0W}} \frac{\tilde Y_{i0}^B}{n_{B0W}}.
\end{eqnarray*}

By the law of large numbers, each of these empirical averages converges in probability to their true means as $n_{B1C},n_{B1W},n_{B0C},n_{B0W} \to \infty$:
\begin{eqnarray*}
    \widehat \pi_{B} &\overset{P}{\rightarrow}& \pi_{B} \\
    n_{B1C}^{-1}\sum_{i=1}^{n_{B1C}} Y_{i1}^B &\overset{P}{\rightarrow}& E[Y_{i1}^B|i\in \Omega_{\bW}^C] \\
    n_{B1W}^{-1}\sum_{i=1}^{n_{B1W}} \tilde Y_{i1}^B &\overset{P}{\rightarrow}& E[ \tilde Y_{i1}^B|i\in \Omega_{\bW}] \\
    n_{B0C}^{-1}\sum_{i=1}^{n_{B0C}} Y_{i0}^B &\overset{P}{\rightarrow}& E[Y_{i0}^B|i\in \Omega_{\bW}^C] \\
    n_{B0W}^{-1}\sum_{i=1}^{n_{B0W}} \tilde Y_{i0}^B &\overset{P}{\rightarrow}& E[ \tilde Y_{i0}^B|i\in \Omega_{\bW}] 
\end{eqnarray*}

Then, by Slutsky's theorem, we have 
\begin{eqnarray*}
    \widehat{\Delta}_{P|A} &\overset{P}{\rightarrow}&  (1-\pi_{B})  E[Y_{i1}^B|i\in \Omega_{\bW}^C] + \pi_{B} E[ \tilde Y_{i1}^B|i\in \Omega_{\bW}] \\ && \hspace{1in} - (1-\pi_{B})  E[Y_{i0}^B|i\in \Omega_{\bW}^C] - \pi_{B} E[ \tilde Y_{i0}^B|i\in \Omega_{\bW}] \\ 
    &=& \pi_B \left\{E[ \tilde Y_{i1}^B|i\in \Omega_{\bW}] - E[ \tilde Y_{i0}^B|i\in \Omega_{\bW}] \right\} \\ && \hspace{1in} + (1-\pi_B) \left\{ E[Y_{i1}^B|i\in \Omega_{\bW}^C] -E[Y_{i0}^B|i\in \Omega_{\bW}^C] \right\} \\
    &=& \pi_B \left \{\int \widehat{\nu}_{A0}(s)dF_{B1|\Omega_{\bW}}(s) - \int \widehat{\nu}_{A0}(s)dF_{B0|\Omega_{\bW}}(s) \right \} \\ && \hspace*{1in} + (1-\pi_B) \left \{ \int y dG_{B1|\Omega_{\bW}^C}(y) - \int y dG_{B0|\Omega_{\bW}^C}(y) \right \} \\ 
    &=& \Delta_{P|A}.
\end{eqnarray*}

Therefore, the proposed estimator is consistent. Next, we will show that $\sqrt{n_B} (\widehat \Delta_{P|A} - \Delta_{P|A})$ converges to a mean 0 normal distribution. Recall that 
\begin{eqnarray*}
    \Delta_{P|A} &=& \pi_B \left \{\int \widehat{\nu}_{A0}(s)dF_{B1|\Omega_{\bW}}(s) - \int \widehat{\nu}_{A0}(s)dF_{B0|\Omega_{\bW}}(s) \right \} \\ && \hspace*{1in} + (1-\pi_B) \left \{ \int y dG_{B1|\Omega_{\bW}^C}(y) - \int y dG_{B0|\Omega_{\bW}^C}(y) \right \} \\  
    &=& \pi_B \int \widehat{\nu}_{A0}(s)dF_{B1|\Omega_{\bW}}(s) - \pi_B \int \widehat{\nu}_{A0}(s)dF_{B0|\Omega_{\bW}}(s) \\ && \hspace*{1in}  + (1-\pi_B) \int y dG_{B1|\Omega_{\bW}^C}(y)  - (1-\pi_B) \int y dG_{B0|\Omega_{\bW}^C}(y).
\end{eqnarray*}

Define:
\begin{eqnarray*}
    I_1 &=&  \int \widehat{\nu}_{A0}(s)dF_{B1|\Omega_{\bW}}(s) \\
    I_2 &=&  \int \widehat{\nu}_{A0}(s)dF_{B0|\Omega_{\bW}}(s) \\
    I_3 &=& \int y dG_{B1|\Omega_{\bW}^C}(y) \\ 
    I_4 &=&  \int y dG_{B0|\Omega_{\bW}^C}(y),
\end{eqnarray*}
such that $\Delta_{P|A} = \pi_B I_1 - \pi_B I_2 + (1-\pi_B) I_3 - (1-\pi_B) I_4$. Note that $I_1, I_2, I_3, I_4$ are defined in different treatment groups in different regions of the covariate space and $\Delta_{P|A}$ conditions on Study A.  We then define the corresponding estimators 
\begin{eqnarray*}
    \widehat I_1 &=& \sum_{i=1}^{n_{B1W}}  \frac{\widehat\nu_{A0}(S_{i1}^B)}{n_{B1W}} \\
    \widehat I_2 &=&  \sum_{i=1}^{n_{B0W}} \frac{\widehat \nu_{A0}(S_{i0}^B)}{n_{B0W}} \\
    \widehat I_3 &=&  \sum_{i=1}^{n_{B1C}} \frac{Y_{i1}^B}{n_{B1C}}  \\
    \widehat I_4 &=&  \sum_{i=1}^{n_{B0C}}\frac{Y_{i0}^B}{n_{B0C}}
\end{eqnarray*}

such that $\widehat{\Delta}_{P|A} = \widehat{\pi}_{B1} \widehat I_1 - \widehat{\pi}_{B0}  \widehat I_2 + (1-\widehat{\pi}_{B1})  \widehat I_3 - (1-\widehat{\pi}_{B0})  \widehat I_4$. By the central limit theorem, 
\begin{eqnarray*}
    \sqrt{n_{B1W}} (\widehat I_1 - I_1) &\overset{d}{\rightarrow}& N(0,\sigma_2^2), \\
    \sqrt{n_{B0W}} (\widehat I_2 - I_2) &\overset{d}{\rightarrow}& N(0,\sigma_4^2), \\
    \sqrt{n_{B1C}} (\widehat I_3 - I_3) &\overset{d}{\rightarrow}& N(0,\sigma_1^2), \\
    \sqrt{n_{B0C}} (\widehat I_4 - I_4) &\overset{d}{\rightarrow}& N(0,\sigma_3^2),
\end{eqnarray*}
where $\sigma_1^2 = \var_{i \in \Omega_{\bW}^C}\left \{ Y_{i1}^B \right \}$, $\sigma_2^2 = \var_{i \in \Omega_{\bW}} \left \{\tilde Y_{i1}^B \right\}$, $\sigma_3^2 = \var_{i \in \Omega_{\bW}^C}\left \{  Y_{i0}^B \right \},$ and $\sigma_4^2 =\var_{i \in \Omega_{\bW}} \left \{\tilde Y_{i0}^B \right \}$. Note that 
\begin{eqnarray*}
\widehat{\pi}_{B1} \widehat I_1 - \pi_B I_1 &=&\widehat{\pi}_{B1} \widehat I_1 - \widehat{\pi}_{B1}  I_1 +\widehat{\pi}_{B1}  I_1  - \pi_B I_1 \\
&=&\widehat{\pi}_{B1} ( \widehat I_1 -  I_1)  + I_1(\widehat{\pi}_{B1}  - \pi_B )
\end{eqnarray*}
and, 
\begin{eqnarray*}
(1-\widehat{\pi}_{B1}) \widehat I_3 - (1-\pi_B) I_3\\
&=& \widehat I_3 -\widehat{\pi}_{B1} \widehat I_3 - I_3 +\pi_B I_3\\
&=& (\widehat I_3 - I_3) - (\widehat{\pi}_{B1} \widehat I_3  -\pi_B I_3)\\
&=& (\widehat I_3 - I_3) - \widehat{\pi}_{B1} ( \widehat I_3 -  I_3)  - I_3(\widehat{\pi}_{B1}  - \pi_B ),
\end{eqnarray*}
and $\widehat{\pi}_{B1}  \overset{p}{\rightarrow} \pi_B$ and $\sqrt{n_{B1}}(\widehat{\pi}_{B1}  - \pi_B) \overset{d}{\rightarrow} N(0, \pi_B(1-\pi_B)).$ Thus, 
\begin{eqnarray*}
\sqrt{n_B}\mathcal{C}_1 &\equiv&\sqrt{n_B} \{ \widehat{\pi}_{B1} \widehat I_1 - \pi_B I_1 + (1-\widehat{\pi}_{B1}) \widehat I_3 - (1-\pi_B) I_3 \} \\
&=& \sqrt{n_B} \{ \widehat{\pi}_{B1} ( \widehat I_1 -  I_1)  + I_1(\widehat{\pi}_{B1}  - \pi_B ) + (\widehat I_3 - I_3) - \widehat{\pi}_{B1} ( \widehat I_3 -  I_3)  - I_3(\widehat{\pi}_{B1}  - \pi_B ) \}\\
&=& \sqrt{n_B} \{\widehat{\pi}_{B1} ( \widehat I_1 -  I_1)   + (1-\widehat{\pi}_{B1})(\widehat I_3 - I_3)  + (I_1-I_3)(\widehat{\pi}_{B1}  - \pi_B ) \}\\
&\overset{d}{\rightarrow}& N(0,p_1^{-1}[\pi_B\sigma_2^2 + (1-\pi_B)\sigma_1^2 + \pi_B(1-\pi_B)(I_1-I_3)^2])
\end{eqnarray*}
By similar arguments, 
\begin{eqnarray*}
\sqrt{n_B}\mathcal{C}_2 &\equiv& \sqrt{n_B} \{\widehat{\pi}_{B0} \widehat I_2 - \pi_B I_2 + (1-\widehat{\pi}_{B0}) \widehat I_4 - (1-\pi_B) I_4 \} \\
&\overset{d}{\rightarrow}& N(0,p_0^{-1}[\pi_B\sigma_4^2 + (1-\pi_B)\sigma_3^2 +  \pi_B(1-\pi_B)(I_2-I_4)^2])
\end{eqnarray*}
And now, 
\begin{eqnarray*}
\sqrt{n_B} (\widehat \Delta_{P|A} - \Delta_{P|A}) &=& \sqrt{n_B} (\mathcal{C}_1 - \mathcal{C}_2)\\
&\overset{d}{\rightarrow}& N(0,\sigma^2_{P|A}) 
\end{eqnarray*}
where \\

\vspace*{-2.5cm}
$$\sigma^2_{P|A}=p_1^{-1}[\pi_B\sigma_2^2 + (1-\pi_B)\sigma_1^2 + \pi_B(1-\pi_B)(I_1-I_3)^2] +  p_0^{-1}[\pi_B\sigma_4^2 + (1-\pi_B)\sigma_3^2 +  \pi_B(1-\pi_B)(I_2-I_4)^2]$$
We estimate $\sigma^2_{P|A}$ plugging in empirical estimates of the corresponding quantities:
\begin{eqnarray*}
\widehat{\sigma}^2_{P|A}&=&\widehat{p}_1^{-1}[\widehat{\pi}_{B1} s_2^2 + (1-\widehat\pi_{B1})s_1^2 + \widehat\pi_{B1}(1-\widehat\pi_{B1})(\bar{y}_{B1W}-\bar{y}_{B1C})^2] \\&&+  \widehat p_0^{-1}[\widehat\pi_{B0} s_4^2 + (1-\widehat \pi_{B0})s_3^2 +  \widehat\pi_{B0}(1-\widehat\pi_{B0})(\bar{y}_{B0W}-\bar{y}_{B0C})^2]
\end{eqnarray*}
where $s_1,s_2,s_3,s_4, \bar{y}_{B0W},\bar{y}_{B0C}, \bar{y}_{B1W},\bar{y}_{B1C}$ are defined in the main text, $\widehat{p}_1 = n_{B1}/n_B$, and $\widehat{p}_0 = n_{B0}/n_B$.

\section*{Appendix C}
Here, we describe in detail the simulation settings. Setting 1 features an extreme case of hetereogeneous surrogate utility, where the surrogate is useless for half of the population ($R_S(W) = 0$ when $W < 5$) and strong for the other half ($R_S(W) = 0.79$ when $W >= 5$). Specifically for both Study A and B, $W \sim U(0,10)$, $S_1 \sim Gamma(shape = 2.55, scale = 2.55)$, $S_0 \sim Gamma(shape = 2.4, scale = 2.4)$. For $W < 5$,
\begin{eqnarray*}
    Y_1 &=& 2.8 + N(0,1) \\
    Y_0 &=& 1 + N(0,1).
\end{eqnarray*}
For $W >= 5$, 
\begin{eqnarray*}
    Y_1 &=& 2.9S_1 + N(0,1) \\
    Y_0 &=& 2.8S_0 + N(0,1).
\end{eqnarray*}
The goal of Setting 2 was to show the performance when the PTE has several possible values rather than two extremes. Specifically for both Study A and B, $W\sim U(0,10)$, $S_1\sim Gamma(shape = 2.55, scale = 2.55)$, $S_0\sim Gamma(shape = 2.4, scale = 2.4)$. For $W < 2.5$,
\begin{eqnarray*}
    Y_1 &=& 2.8 + N(0,3) \\ 
    Y_0 &=& 1 + N(0,3) \\ 
    &\implies& R_S(W) = 0.
\end{eqnarray*}
For $2.5 \leq W < 5$,
\begin{eqnarray*}
    Y_1 &=& 1.1 + 0.4 S_1 + N(0,3) \\ 
    Y_0 &=& 0.8 + 0.3 S_0 + N(0,3) \\ 
    &\implies& R_S(W) = 0.25.
\end{eqnarray*}
For $5 \leq W < 7.5$,
\begin{eqnarray*}
    Y_1 &=& 1.5 + 1.6 S_1 + N(0,3) \\ 
    Y_0 &=& 1 + 1.5 S_0 + N(0,3) \\ 
    &\implies& R_S(W) = 0.52.
\end{eqnarray*}
For $7.5 \leq W$, 
\begin{eqnarray*}
    Y_1 &=& 1.85 S_1 + N(0,3) \\
    Y_0 &=& 1.8 S_0 + N(0,3) \\
    &\implies& R_S(W) = 0.83.
\end{eqnarray*}

In setting 3, there was no treatment effect and thus $R_S(W)$ is undefined. Specifically, in both study A and B, $W\sim U(0,12)$, $S_1= S_0 = N(2, 3^2)$, $Y_1 = Y_0 = 2S + W + N(0,6^2)$.

\clearpage
\appendix

\setcounter{table}{0}
\renewcommand{\thetable}{A\arabic{table}}
\setcounter{figure}{0}
\renewcommand{\thefigure}{A\arabic{figure}}
\renewcommand{\theequation}{A.\arabic{equation}}

\allowdisplaybreaks

\section*{Appendix A}
 \setlength{\parindent}{0cm}
\textit{Proof of Theorem 1.} Here, we will show that 
$$\sqrt{n}\left(\begin{array}{c}\widehat{\Delta}(\bw)-\Delta(\bw) \\ \widehat{\Delta}_S(\bw)-\Delta_S(\bw)\end{array}\right)=\frac{1}{\sqrt{n_1}}\sum_{i=1}^{n_1} \left(\begin{array}{c} \xi_{1i}(\bw) \\ \xi_{2i}(\bw) \end{array}\right)+\frac{1}{\sqrt{n_0}}\sum_{j=1}^{n_0} \left(\begin{array}{c} \zeta_{1j}(\bw) \\ \zeta_{2j}(\bw) \end{array}\right)+o_p(1),$$
where $\{\xi_{1i}(\bw), \xi_{2i}(\bw)\}'$ are independent and identically distributed (iid) mean zero random vector, $\{\zeta_{1j}(\bw), \zeta_{2j}(\bw)\}'$ are also iid mean zero random vectors, and thus converges weakly to a mean zero bivariate normal distribution with a variance-covariance matrix of $\Sigma_{\Delta}(\bw)$
by the central limit theorem. First, note that 
\begin{eqnarray*}
&& \sqrt{n}(\widehat{\Delta}(\bw) - \Delta(\bw)) \\ &=&  \sqrt{n}\{(\widehat \beta_1 + (\widehat \beta_2 + \widehat \beta_3) \widehat \alpha_1 + \boldsymbol{\widehat\beta}_5\trans \bw - \widehat \beta_2  \widehat \alpha_0) - (\beta_1 + (\beta_2+\beta_3)\alpha_1 + \boldsymbol{\beta}_5\trans \bw - \beta_2 \alpha_0) \} \\
&=& \sqrt{n}\{(\widehat \beta_1 - \beta_1) + ((\widehat \beta_2 + \widehat \beta_3) \widehat \alpha_1 - (\beta_2+\beta_3)\alpha_1 )+ (\boldsymbol{\widehat\beta}_5\trans \bw - \boldsymbol{\beta}_5\trans \bw) - (\widehat \beta_2  \widehat \alpha_0 - \beta_2 \alpha_0) \} \\
&=& \sqrt{n}\{(\widehat \beta_1 - \beta_1) + (\boldsymbol{\widehat\beta}_5\trans \bw - \boldsymbol{\beta}_5\trans \bw)  + ((\widehat \beta_2 + \widehat \beta_3) \widehat \alpha_1 - (\beta_2+\beta_3)\alpha_1 )  \\ &&- (\widehat \beta_2  \widehat \alpha_0 - \beta_2 \alpha_0)- \alpha_0(\widehat \beta_2 - \beta_2)  + \alpha_0(\widehat \beta_2 - \beta_2)  - \beta_2(\widehat \alpha_0 - \alpha_0)  + \beta_2(\widehat \alpha_0 - \alpha_0)   \}\\
&=& \sqrt{n}\{(\widehat \beta_1 - \beta_1) + (\boldsymbol{\widehat\beta}_5\trans \bw - \boldsymbol{\beta}_5\trans \bw) - \alpha_0(\widehat \beta_2 - \beta_2) - \beta_2(\widehat \alpha_0 - \alpha_0)  \\ &&
+ ((\widehat \beta_2 + \widehat \beta_3) \widehat \alpha_1 - (\beta_2+\beta_3)\alpha_1 ) - \widehat \beta_2  \widehat \alpha_0 + \beta_2 \alpha_0  + \alpha_0\widehat \beta_2 - \alpha_0\beta_2  +\beta_2\widehat \alpha_0 - \beta_2 \alpha_0\}\\
&=& \sqrt{n}\{(\widehat \beta_1 - \beta_1) + (\boldsymbol{\widehat\beta}_5\trans \bw - \boldsymbol{\beta}_5\trans \bw) - \alpha_0(\widehat \beta_2 - \beta_2) - \beta_2(\widehat \alpha_0 - \alpha_0)  \\ &&
+ ((\widehat \beta_2 + \widehat \beta_3) \widehat \alpha_1 - (\beta_2+\beta_3)\alpha_1 ) - \widehat \beta_2  \widehat \alpha_0 + \alpha_0\widehat \beta_2  +\beta_2\widehat \alpha_0 - \beta_2 \alpha_0\}\\
&=& \sqrt{n}\{(\widehat \beta_1 - \beta_1) + (\boldsymbol{\widehat\beta}_5\trans \bw - \boldsymbol{\beta}_5\trans \bw) - \alpha_0(\widehat \beta_2 - \beta_2) - \beta_2(\widehat \alpha_0 - \alpha_0)  \\ &&
+ ((\widehat \beta_2 + \widehat \beta_3) \widehat \alpha_1 - (\beta_2+\beta_3)\alpha_1 ) - (\widehat \beta_2-  \beta_2 )( \widehat \alpha_0 - \alpha_0) \}\\
&=& \sqrt{n}\{(\widehat \beta_1 - \beta_1) + (\boldsymbol{\widehat\beta}_5\trans \bw - \boldsymbol{\beta}_5\trans \bw) - \alpha_0(\widehat \beta_2 - \beta_2) - \beta_2(\widehat \alpha_0 - \alpha_0)  \\ &&
+ ((\widehat \beta_2 + \widehat \beta_3) \widehat \alpha_1 - (\beta_2+\beta_3)\alpha_1 ) \} + o_p(1)\\
&=& \sqrt{n}\{(\widehat \beta_1 - \beta_1) + (\boldsymbol{\widehat\beta}_5\trans \bw - \boldsymbol{\beta}_5\trans \bw) - \alpha_0(\widehat \beta_2 - \beta_2) - \beta_2(\widehat \alpha_0 - \alpha_0) \\ &&
+ \alpha_1((\widehat \beta_2 + \widehat \beta_3) - (\beta_2+\beta_3)) + (\beta_2 + \beta_3)( \widehat \alpha_1 - \alpha_1 ) \} + o_p(1)\\
&=& \sqrt{n}\{(\widehat \beta_1 - \beta_1) + (\boldsymbol{\widehat\beta}_5\trans \bw - \boldsymbol{\beta}_5\trans \bw) - \alpha_0(\widehat \beta_2 - \beta_2) + \alpha_1((\widehat \beta_2 + \widehat \beta_3) - (\beta_2+\beta_3))  \\ &&
+ (\beta_2 + \beta_3)( \widehat \alpha_1 - \alpha_1 ) - \beta_2(\widehat \alpha_0 - \alpha_0)\} + o_p(1)\\
&=& \sqrt{n}\{(\widehat \beta_1 - \beta_1)  + (\alpha_1 - \alpha_0)(\widehat \beta_2 - \beta_2) + \alpha_1(\widehat \beta_3 - \beta_3)  + (\boldsymbol{\widehat\beta}_5\trans \bw - \boldsymbol{\beta}_5\trans \bw)\\ &&
+ (\beta_2 + \beta_3)( \widehat \alpha_1 - \alpha_1 ) - \beta_2(\widehat \alpha_0 - \alpha_0)\} + o_p(1)
\end{eqnarray*}

Thus, using matrix operation, we may express this quantity as:
$$
\sqrt{n}(\widehat{\Delta}(\bw) - \Delta(\bw))=\mathbf{B}_0(\bw)\trans\left(\begin{array}{c} \sqrt{n}(\widehat{\beta}_1-\beta_1)\\ \sqrt{n}(\widehat{\beta}_2-\beta_2) \\ \sqrt{n}(\widehat{\beta}_3-\beta_3)\\ \sqrt{n}(\widehat{\beta}_5-\beta_5) \\
                            \sqrt{n}(\widehat{\alpha}_1-\alpha_1)\\ \sqrt{n}(\widehat{\alpha}_0-\alpha_0)
                            \end{array}\right)+o_p(1)
$$
where $\mathbf{B}_0(\bw)=(1, \alpha_1-\alpha_0, \alpha_1, \bw\trans, \beta_2+\beta_3, -\beta_2)\trans$.
Now, let $\mathbf{0}_p$ be a $p$-dimensional vector consisting of all zeros, and $\mathbf{0}_{p\times p}$ be a $p$ by $p$ matrix consisting of all zero entries. It can be shown that 
\begin{eqnarray*}
\sqrt{n}\left(\begin{array}{c} (\widehat{\beta}_1-\beta_1)\\ (\widehat{\beta}_2-\beta_2) \\ (\widehat{\beta}_3-\beta_3)\\ (\widehat{\beta}_5-\beta_5) \\
                            (\widehat{\alpha}_1-\alpha_1)\\ (\widehat{\alpha}_0-\alpha_0)
                            \end{array}\right)%&=&\sqrt{n} \mathbf{C}_0\mathbf{A}_0^{-1}\frac{1}{n}\sum_{i=1}^n \left(\begin{array}{c} G_i\left[Y_{1i}-(\beta_0+\beta_1)-(\beta_2+\beta_3)S_{1i}-(\beta_4+\beta_5)\trans\bW_{1i}\right] \\ 
                            %G_iS_{1i}\left[Y_{1i}-(\beta_0+\beta_1)-(\beta_2+\beta_3)S_{1i}-(\beta_4+\beta_5)\trans\bW_{1i}\right] \\ 
                            %G_i\bW_{1i}\left[Y_{1i}-(\beta_0+\beta_1)-(\beta_2+\beta_3)S_{1i}-(\beta_4+\beta_5)\trans\bW_{1i}\right] \\ 
                            %(1-G_i)\left[Y_{0i}-\beta_0-\beta_2S_{0i}-\beta_4\trans\bW_{0i}\right]\\
                            %(1-G_i)S_{0i}\left[Y_{0i}-\beta_0-\beta_2S_{0i}-\beta_4\trans\bW_{0i}\right]\\
                            %(1-G_i)\bW_{0i}\left[Y_{0i}-\beta_0-\beta_2S_{0i}-\beta_4\trans\bW_{0i}\right]\\
                            %G_i(S_{1i}-\alpha_1)\\
                            %(1-G_i)(S_{0i}-\alpha_0)
                            %\end{array}\right)\\
                       &=&\mathbf{C}_0\mathbf{A}_0^{-1}\left[\frac{1}{\sqrt{n_1}}\sum_{i=1}^{n_1}      \sqrt{\pi_1} \left(\begin{array}{c} \epsilon_{1i} \\ 
                            S_{1i}\epsilon_{1i} \\ 
                            \bW_{1i}\epsilon_{1i} \\ 
                            0\\
                            0\\
                            \mathbf{0}_p\\
                            S_{1i}-\alpha_1\\
                            0
                            \end{array}\right) + \frac{1}{\sqrt{n_0}}\sum_{j=1}^{n_0} \sqrt{\pi_0} \left(\begin{array}{c} 0 \\ 
                            0 \\ 
                            \mathbf{0}_p \\ 
                            \epsilon_{0j}\\
                            S_{0i}\epsilon_{0j}\\
                           \bW_{0i}\epsilon_{0j}\\
                            0\\
                            S_{0i}-\alpha_0
                            \end{array}\right)\right\}\\
                            &=&\mathbf{C}_0\mathbf{A}_0^{-1}\left[\frac{1}{\sqrt{n_1}}\sum_{i=1}^{n_1} \mathbf{u}_{1i} + \frac{1}{\sqrt{n_0}}\sum_{j=1}^{n_0} \mathbf{v}_{0j}\right]
             \end{eqnarray*}         
where $\epsilon_{1i}=Y_{1i}-(\beta_0+\beta_1)-(\beta_2+\beta_3)S_{1i}-(\beta_4+\beta_5)\trans\bW_{1i},$ $\epsilon_{0j}=Y_{0j}-\beta_0-\beta_2S_{0j}-\beta_4\trans\bW_{0j},$
$$\mathbf{C}_0=\left(\begin{array}{cccccccc}
0 & 1 & 0 & 0 &  \mathbf{0}_p\trans &  \mathbf{0}_p\trans & 0 & 0\\
0 & 0 & 1 & 0 &  \mathbf{0}_p\trans &  \mathbf{0}_p\trans & 0 & 0\\
0 & 0 & 0 & 1 &  \mathbf{0}_p\trans &  \mathbf{0}_p\trans & 0 & 0\\
0 & 0 & 0 & 0 &  \mathbf{0}_{p\times p} &  \mathbf{I}_p & 0 & 0\\
0 & 0 & 0 & 0 &  \mathbf{0}_p\trans &  \mathbf{0}_p\trans & 1 & 0\\
0 & 0 & 0 & 0 &  \mathbf{0}_p\trans &  \mathbf{0}_p\trans & 0 & 1
\end{array}\right)$$
and
$$\mathbf{A}_0=E\left(\begin{array}{cccccccc}    
\pi_1 &  \pi_1S^{(1)}  & \pi_1\bW\trans  & \pi_0  & \pi_0S^{(0)}  & \pi_0\bW\trans & 0 & 0\\
\pi_1 &  \pi_1S^{(1)}  & \pi_1\bW\trans  &  0     &   0     &  \mathbf{0}_p\trans       & 0 & 0\\
\pi_1S^{(1)} & \pi_1(S^{(1)})^2 & \pi_1S^{(1)}\bW\trans & \pi_0S^{(0)} & \pi_0(S^{(0)})^2 & \pi_0S^{(0)}\bW\trans & 0 & 0\\
\pi_1S^{(1)} & \pi_1(S^{(1)})^2 & \pi_1S^{(1)}\bW\trans & 0 &  0 & \mathbf{0}_p\trans & 0 & 0\\
\pi_1\bW & \pi_1S^{(1)}\bW & \pi_1\bW\bW\trans & \pi_0\bW & \pi_0S^{(0)}\bW & \pi_0\bW\bW\trans & 0 & 0\\
\pi_1\bW & \pi_1S^{(1)}\bW & \pi_1\bW\bW\trans & \mathbf{0}_p & \mathbf{0}_p & \mathbf{0}_{p\times p} & 0 & 0\\
0 & 0 & \mathbf{0}_p\trans & 0 & 0 & \mathbf{0}_p\trans & \pi_1 & 0\\
0 & 0 & \mathbf{0}_p\trans & 0 & 0 & \mathbf{0}_p\trans & 0 & \pi_0
\end{array}\right).$$

\begin{comment}
\textcolor{red}{[Lu - I am terrible with matrices, but here is a question, I just tried to check my sanity by multiplying this all through for one term, the second to last one. So $\sqrt{n}(\hat{\alpha}_1 - \alpha_1)$ and when I multiply through I get (let $\pi_1 = P(G=1)$):
\begin{eqnarray*}
    &&\frac{1}{n}\sum_{i=1}^n G_i(S_i - \alpha_1)E(G)\\
    &=&  \frac{1}{n}\sum_{i=1}^n G_iS_iE(G) - \frac{1}{n}\sum_{i=1}^n G_i\alpha_1E(G)\\
    &=&  \frac{1}{n}\sum_{i=1}^{n_1} S_i\pi_1 - \frac{1}{n}n_1 \pi_1\alpha_1\\
    &\approx& \pi_1(\frac{1}{n_1}\sum_{i=1}^{n_1} S_i)\pi_1 - \pi_1 \pi_1\alpha_1\\
\end{eqnarray*}
But don't I need the $\pi_1$'s to cancel? I was expecting one to be in the denominator so I can cancel. Or did I do something dumb in my algebra here?]}
\textcolor{red}{[Lu, is the statement below right?]}
\end{comment}
Therefore, $$\sqrt{n}(\widehat{\Delta}(\bw) - \Delta(\bw))=\frac{1}{\sqrt{n_1}}\sum_{i=1}^{n_1} \mathbf{B}_0(\bw)\trans\mathbf{C}_0\mathbf{A}_0^{-1}\mathbf{u}_{1i} + \frac{1}{\sqrt{n_0}}\sum_{j=1}^{n_0} \mathbf{B}_0(\bw)\trans\mathbf{C}_0\mathbf{A}_0^{-1}\mathbf{v}_{0j} +o_p(1)$$
where $\mathbf{u}_{1i}$ and $\mathbf{v}_{0j}$ are iid mean zero random vectors. Similarly, 
\begin{align*}
    &\sqrt{n}(\widehat{\Delta}_S(\bw)-\Delta_S(\bw))=\mathbf{B}_1(\bw)\trans \left(\begin{array}{c} \sqrt{n}(\widehat{\beta}_1-\beta_1)\\ \sqrt{n}(\widehat{\beta}_3-\beta_3) \\ \sqrt{n}(\widehat{\beta}_5-\beta_5)\end{array} \right)\\
    =&\frac{1}{\sqrt{n_1}}\sum_{i=1}^{n_1} \mathbf{B}_1(\bw)\trans\mathbf{C}_1\mathbf{A}_0^{-1}\mathbf{u}_{1i} + \frac{1}{\sqrt{n_0}}\sum_{j=1}^{n_0} \mathbf{B}_1(\bw)\trans\mathbf{C}_1\mathbf{A}_0^{-1}\mathbf{v}_{0j} +o_p(1),
\end{align*}
where $\mathbf{B}_1(\bw)=(1,1, \bw\trans)\trans$ and 
$$ \mathbf{C}_1=\left(\begin{array}{cccccccc}
0 & 1 & 0 & 0 &  \mathbf{0}_p\trans &  \mathbf{0}_p\trans & 0 & 0\\
0 & 0 & 0 & 1 &  \mathbf{0}_p\trans &  \mathbf{0}_p\trans & 0 & 0\\
0 & 0 & 0 & 0 &  \mathbf{0}_{p\times p} &  \mathbf{I}_p & 0 & 0
\end{array}\right).$$

Therefore, as $n \rightarrow \infty, $ by the central limit theorem, $$\sqrt{n}\left(\begin{array}{c}\widehat{\Delta}(\bw)-\Delta(\bw) \\ \widehat{\Delta}_S(\bw)-\Delta_S(\bw)\end{array}\right)=\frac{1}{\sqrt{n_1}}\sum_{i=1}^{n_1} \left(\begin{array}{c} \mathbf{B}_0(\bw)\trans\mathbf{C}_0 \\ \mathbf{B}_1(\bw)\trans\mathbf{C}_1\end{array}\right)\mathbf{A}_0^{-1}\mathbf{u}_{1i}+\frac{1}{\sqrt{n_0}}\sum_{j=1}^{n_0} \left(\begin{array}{c} \mathbf{B}_0(\bw)\trans\mathbf{C}_0 \\ \mathbf{B}_1(\bw)\trans\mathbf{C}_1\end{array}\right)\mathbf{A}_0^{-1}\mathbf{v}_{0j}+o_p(1)$$
converges weakly to a mean zero bivariate normal distribution with variance-covariance matrix, 
$$\Sigma_{\Delta}(\bw) =\left(\begin{array}{c} \mathbf{B}_0(\bw)\trans\mathbf{C}_0 \\ \mathbf{B}_1(\bw)\trans\mathbf{C}_1\end{array}\right)\mathbf{A}_0^{-1} E \left(\mathbf{u}^{\otimes 2}+\mathbf{v}^{\otimes 2}\right) \mathbf{A}_0^{-1} (\mathbf{C}_0\trans\mathbf{B}_0(\bw), \mathbf{C}_1\trans\mathbf{B}_1(\bw)) $$
where $\mathbf{a}^{\otimes 2} = \mathbf{a}\mathbf{a}^{\transpose}$ for a vector $\mathbf{a}$. By the delta method, $\sqrt{n}\left\{\widehat{R}_S(\bw) - R_S(\bw)\right\}$ also converges weakly to a mean zero normal distribution as $n \rightarrow \infty $, if $\Delta(\bw) \neq 0$.

\section*{Appendix B}
\textit{Proof of Theorem 2.} %We require the following conditions: (a) Assumptions (C1)-(C5) hold, (b) $S \supone$ and $S \supzero$ have the same bounded support on $\Omega_S$, [\textcolor{blue}{[Do we need something about the density functions of S or U having continuous derivatives bounded away from zero?]} (c) all bandwidths $h_k$ are such that $h_k = O_p(n_g^{-\epsilon}), 1/5 < \epsilon < 1/2$ for $k=0,1,2,3,4$, (d) $h_k/h \in [r_L,r_u]$ for finite positive constants, (e) $\pi_g = \lim_{n\rightarrow \infty} n_g/n \in (0,1), g = 0,1,$ $\pi_g = \lim_{n\rightarrow \infty} n_g/n \in (0,1), g = 0,1$, (f) $\widehat{U}(\bw)\rightarrow u_0(\bw)$, and (g) 
%\begin{equation*}
%    \Delta(\bw), \Delta_S(\bw) \perp \bW \mid u_0(\bW)
%\end{equation*}
%holds. Here, 
We must show that: 
\begin{align}
&\sqrt{nh}\left(\begin{array}{c}\widehat{\Delta}^K(\bw)-\Delta^K(\bw) \\ \widehat{\Delta}_{S}^K(\bw)-\Delta_{S}(\bw)\end{array}\right)\\
=&\frac{\sqrt{h}}{\sqrt{\pi_1n_1}}\sum_{i=1}^{n_1} \left(\begin{array}{c} \xi_{1i}(u_0(\bw)) \\ \xi_{2i}(u_0(\bw)) \end{array}\right)+\frac{\sqrt{h}}{\sqrt{\pi_0n_0}}\sum_{i=1}^{n_0} \left(\begin{array}{c} \zeta_{1i}(u_0(\bw)) \\ \zeta_{2i}(u_0(\bw)) \end{array}\right)+o_p(1),
\label{eq:firstoderapp}
\end{align}

under the regularity conditions in Theorem 2,  where $\xi_{1i}$ are iid terms, $\xi_{2i}$ are iid terms.
Coupled with the central limit theorem, it would follow that $\sqrt{nh}\left(\widehat{\Delta}^K(\bw)-\Delta^K(\bw)\right)$ and $\sqrt{nh}\left(\widehat{\Delta}_{S}^K(\bw)-\Delta_{S}(\bw)\right)$ converge weakly to a mean zero bivariate normal distribution with a variance-covariance matrix of $\Sigma_{K}(\bw)$. Lastly, by the delta method, $\widehat{R}^K_S(\bw)$ would converge weakly to a mean zero normal distribution as $n \rightarrow \infty $ as long as $\Delta^K(\bw) \neq 0$.  Without loss of generality, we assume that $h_0=h_1=h_2=h_3=h_4=h$ in the following. To demonstrate the first order approximation (\ref{eq:firstoderapp}), we  first show that  
\begin{equation}\widehat{\Delta}^K(\bw) -\widetilde{\Delta}^K(\bw)=o_p\left(\frac{1}{\sqrt{nh}} \right),\label{eq:appx1}\end{equation}
where 
$$\widetilde{\Delta}^K(\bw) = \widetilde{m}_1\left\{u_0(\bw)\right\} - \widetilde{m}_0\left\{u_0(\bw)\right\},$$
$$\widetilde{m}_g(u) = \frac{ \sum_{i=1}^{n_g} K_{h_g}(U_{gi} - u)Y_{gi}}{\sum_{i=1}^{n_g}K_{h_g}(U_{gi} - u)}, g\in \{0, 1\},$$
and $U_{gi}=u_0(\bW_{gi})=u(\bW_{gi}, \theta_0).$  Now, consider 
\begin{align*}
&\frac{1}{n_g}\left\{\sum_{i=1}^{n_g}K_{h_g}(\widehat{U}_{gi} - u)-\sum_{i=1}^{n_g}K_{h_g}(U_{gi} - u)\right\}\\
=&\frac{1}{n_gh_g}\sum_{i=1}^{n_g}\left[K\left(\frac{\widehat{U}_{gi} - u}{h_g}\right)-K\left(\frac{U_{gi} - u}{h_g}\right)\right]\\
=&\frac{1}{n_gh_g}\sum_{i=1}^{n_g} \left\{\dot{K}\left(\frac{U_{gi} - u}{h_g}\right)\left(\frac{\widehat{U}_{gi}-U_i}{h_g}\right)\right\}(1+o_p(1))\\
=&\frac{1}{n_gh_g}\sum_{i=1}^{n_g} \left\{\dot{K}\left(\frac{U_{gi} - u}{h_g}\right)\left(\frac{g(\bW_i)(\hat{\theta}-\theta_0)}{h_g}\right)\right\}(1+o_p(1))\\
=&\left[\frac{1}{n_g}\sum_{i=1}^{n_g}\dot{K}_{h_g}(U_{gi} - u)g(\bW_{gi})\right]\times O_p\left(\frac{1}{\sqrt{n_g}h_g}\right)\\
=& O_p\left(h_g+\frac{\log(n_g)\sqrt{h_g}}{\sqrt{n_g}} \right)O_p\left(\frac{1}{\sqrt{n_g}h_g}\right),
\end{align*}
which is $o_p\{(n_gh_g)^{-1/2}\}$ uniformly in $u,$ where $g(\bw)=\partial u(\bw, \theta)/\partial \theta|_{\theta=\theta_0}.$
Next, consider 
\begin{align*}
&\frac{1}{n_g}\left\{\sum_{i=1}^{n_g}K_{h_g}(\widehat{U}_{gi} - u)Y_i-\sum_{i=1}^{n_g}K_{h_g}(U_{gi} - u)Y_{gi}\right\}\\
=&\left[\frac{1}{n_g}\sum_{i=1}^{n_g}\dot{K}_{h_g}(U_{gi} - u)g(\bW_{gi})Y_{gi}\right]\times O_p\left(\frac{1}{\sqrt{n_g}h_g}\right)\\
=& O(h_g)O_p\left(1+\frac{\log(n_g)}{\sqrt{n_gh_g}} \right)O_p\left(\frac{1}{\sqrt{n_g}h_g}\right)=o_p\left( \frac{1}{\sqrt{n_gh_g}}\right).
\end{align*}
Therefore, 
$$\widetilde{m}_g(u)-\widehat{m}_g(u)=o_p\left( \frac{1}{\sqrt{n_gh_g}}\right).$$
Next 
\begin{align*}
&\widehat{\Delta}^K(\bw) -\widetilde{\Delta}^K(\bw)\\
=&\left[\widehat{m}_1\left\{u(\bw, \hat{\theta})\right\}-\widetilde{m}_1\left\{u_0(\bw)\right\}\right] - \left[\widehat{m}_0\left\{u(\bw, \hat{\theta})\right\}-\widetilde{m}_0\left\{u_0(\bw)\right\}\right]\\
\le & \left|\widehat{m}_1\left\{u(\bw, \hat{\theta})\right\}-\widetilde{m}_1\left\{u(\bw, \hat{\theta})\right\}\right| + \left|\widetilde{m}_1\left\{u(\bw, \hat{\theta})\right\}-\widetilde{m}_1\left\{u_0(\bw)\right\}\right|\\ &+\left|\widehat{m}_0\left\{u(\bw, \hat{\theta})\right\}-\widetilde{m}_0\left\{\hat{u}(\bw, \hat{\theta})\right\}\right|+\left|\widetilde{m}_0\left\{u(\bw, \hat{\theta})\right\}-\widetilde{m}_0\left\{u_0(\bw)\right\}\right|\\
=&o_p\left( \frac{1}{\sqrt{n_gh_g}}\right)+\left|\widetilde{m}_1\left\{u(\bw, \hat{\theta})\right\}-\widetilde{m}_1\left\{u_0(\bw)\right\}\right|+\left|\widetilde{m}_0\left\{u(\bw, \hat{\theta})\right\}-\widetilde{m}_0\left\{u_0(\bw)\right\}\right|\\
=&o_p\left( \frac{1}{\sqrt{n_gh_g}}\right)+O_p\left(\frac{1}{n_g}\right),
\end{align*}
where we used the fact that 
$\sup_{\bw}\left|u(\bw, \hat{\theta})-u_0(\bw)\right|=O_p(n_g^{-1/2})$
and 
$$ \sup_u \left|\frac{d \widetilde{m}_g(u)}{du}-\frac{d m_g(u)}{du}\right|=O_p\left(\frac{\log(n)}{\sqrt{n_gh_g^3}}+  h_g\right)=o_p(1),$$
for $\epsilon\in (0, 1/3).$  Therefore, (\ref{eq:appx1}) is established. Next, we need to establish that 
\begin{equation}\widehat{\Delta}_S^K(\bw) -\widetilde{\Delta}_S^K(\bw)=o_p\left(\frac{1}{\sqrt{nh}} \right),\label{eq:appx2}\end{equation}
where
$$\widetilde{\Delta}^K_{S}(\bw) = \widetilde{m}_{10}(u_0(\bw))-\widetilde{m}_0(u_0(\bw)),$$
 where $\widetilde{m}_{10}(u)=\int \widetilde{\mu}_{1}(s, u)d\widetilde{F}_{S^{(0)}}(s\mid u),$
 \begin{align*}
  \widetilde{F}_{S^{(0)}}(s \mid u)=\frac{\sum_{i=1}^{n_0}K_{h_2}(U_{0i}-u)I(S_{0i}\le s)}{\sum_{i=1}^{n_0}K_{h_2}(U_{0i}-u)}~~
\mbox{and}~~\widetilde{\mu}_{1}(s, u)= \frac{\sum_{i=1}^{n_1}K_{h_3}(S_{1i}-s)K_{h_4}(U_{1i}-u)Y_{1i}}{\sum_{i=1}^{n_1}K_{h_3}(S_{1i}-s)K_{h_4}(U_{1i}-u)}.
  \end{align*}
Since we have already established that 
$$\widehat{m}_0\left\{u(\bw, \hat{\theta})\right\}-\widetilde{m}_0(u_0(\bw))=o_p\left(\frac{1}{\sqrt{n_0h_0}} \right)$$
uniformly in $\bw$, we only need to show that 
$$ \widehat{m}_{10}\left\{u(\bw, \hat{\theta})\right\}-\widetilde{m}_{10}(u_0(\bw))=o_p\left(\frac{1}{\sqrt{n_0h_0}} \right).$$
Consider 
\begin{align*}
 &\frac{1}{n_1h_3}\sum_{i=1}^{n_1} \left\{\dot{K}_{h_3}(S_{1i}-s)K_{h_4}(\hat{U}_{1i}-u)-\dot{K}_{h_3}(S_{1i}-s)K_{h_4}(U_{1i}-u)\right\}\\
 =&\frac{1}{n_1h_3}\sum_{i=1}^{n_1}\left[K\left(\frac{\widehat{U}_{1i} - u}{h_4}\right)-K\left(\frac{U_{1i} - u}{h_4}\right)\right]\dot{K}_{h_3}(S_{1i}-s)\\
=&\frac{1}{n_1h_3}\sum_{i=1}^{n_1} \left\{\dot{K}\left(\frac{U_{1i} - u}{h_4}\right)\left(\frac{\widehat{U}_{1i}-U_{1i}}{h_4}\right)\right\}\dot{K}_{h_3}(S_{1i}-s)(1+o_p(1))\\
=&\frac{1}{n_1h_3^2}\sum_{i=1}^{n_1} \left\{\dot{K}\left(\frac{U_{1i} - u}{h_4}\right)\dot{K}\left(\frac{S_{1i} - s}{h_3}\right)\left(\frac{g(\bW_{1i})(\hat{\theta}-\theta_0)}{h_4}\right)\right\}(1+o_p(1))\\
=&\left[\frac{1}{n_1}\sum_{i=1}^{n_1}\dot{K}_{h_4}(U_{1i} - u)\dot{K}_{h_3}(S_{1i}-s)g(\bW_i)\right]\times O_p\left(\frac{1}{\sqrt{n_1}h_3}\right)\\
=& O(h_3h_4)O_p\left(1+\frac{\log(n_1)}{\sqrt{n_1h_3h_4}} \right)O_p\left(\frac{1}{\sqrt{n_1}h_3}\right),
\end{align*}
which is $o_p\{(n_1h_3)^{-1/2}\}$ uniformly in $u,$ if $n_1h_3h_4/\log(n) \rightarrow \infty.$  Following the same arguments above, it is straightforward to show that 
$$\widehat{F}_{S^{(0)}}(s\mid u)-\widehat{F}_{S^{(0)}}(s\mid u)=o_p\left(\frac{1}{\sqrt{n_0h_2}} \right).$$
Therefore, 
\begin{align*}
 &\widehat{m}_{10}(u)-\widetilde{m}_{10}(u)\\
=& \int \widehat{\mu}_1(s, u)d\widehat{F}_{S^{(0)}}(s\mid u)-\int \widetilde{\mu}_1(s, u)d\widetilde{F}_{S^{(0)}}(s\mid u)\\
=&   \int \widehat{F}_{S^{(0)}}(s\mid u)\frac{\partial \widehat{\mu}_1(s, u)}{\partial s}ds-\int \widetilde{F}_{S^{(0)}}(s\mid u) \frac{\partial \widetilde{\mu}_1(s, u)}{\partial s}ds\\
=& \int \widehat{F}_{S^{(0)}}(s\mid u)\left[\frac{\partial \widehat{\mu}_1(s, u)}{\partial s}- \frac{\partial \widetilde{\mu}_1(s, u)}{\partial s}\right]ds + \int \left[\widehat{F}_{S^{(0)}}(s\mid u)-\widetilde{F}_{S^{(0)}}(s\mid u)\right] \frac{\partial \widetilde{\mu}_1(s, u)}{\partial s}ds\\
=& o_p\left(\frac{1}{\sqrt{n_0h_2}} +\frac{1}{\sqrt{n_1h_3}}\right).
\end{align*}
Coupled with the fact that $U(\bw, \hat{\theta})-u_0(\bw)=O_p(n^{-1/2})$, it implies that (\ref{eq:appx2}). Therefore, we have established that under the condition of Theorem 2
$$ \sqrt{nh}\left(\begin{array}{c}\widehat{\Delta}^K(\bw)-\widetilde{\Delta}^K(\bw) \\ \widehat{\Delta}_{S}^K(\bw)-\widetilde{\Delta}_{S}(\bw)\end{array}\right)=o_p(1).$$
Based on the result by \cite{parast2023testing}, we have
\begin{equation}
\small
\sqrt{nh}\left(\begin{array}{c}\widetilde{\Delta}^K(u)-\Delta(u) \\ \widetilde{\Delta}_{S}^K(u)-\Delta_{S}(u)\end{array}\right)=\frac{\sqrt{h}}{\sqrt{\pi_1n_1}}\sum_{i=1}^{n_1} \left(\begin{array}{c} \xi_{1i}(u) \\ \xi_{2i}(u) \end{array}\right)+\frac{\sqrt{h}}{\sqrt{\pi_0n_0}}\sum_{i=1}^{n_0} \left(\begin{array}{c} \zeta_{1i}(u) \\ \zeta_{2i}(u) \end{array}\right)+o_p(1),
\end{equation}
where $u=u_0(\bw),$ $\pi_g=n_g/n,$
\begin{align*}
\xi_{1i}(u)&=\frac{K_{h_1}(U_{1i}-u)}{f_U(u)}\\
\zeta_{1i}(u)&=-\frac{K_{h_0}(U_{0i}-u)}{f_U(u)}\\
\xi_{2i}(u)&=\left[\frac{K_{h_2}(U_{0i}-u)}{f_U(u)}\left\{\mu_1(S_{0i}, U_{0i})-\mu_{10}(U_{0i}) \right\}-\frac{K_{h_0}(U_{0i}-u)}{f_U(u)}\left\{Y_{0i}-\mu_0(U_{0i})\right\} \right]\\
\zeta_{2i}(u)&=\frac{K_{h_4}(U_{1i}-u)}{f_U(u)}\left\{Y_{1i}-\mu_1(S_{1i}, U_{1i})\right\} \frac{f_{S^{(0)}}(S_{1i} \mid u)}{f_{S^{(1)}}(S_{1i}\mid u)},
\end{align*}
$f_U(\cdot)$ is the density function of $U=u_0(\bW),$ $\mu_1(s, u)=E(Y^{(1)}\mid S^{(1)}=s, U=u),$ $\mu_g(u)=E(Y^{(g)}\mid U=u),$ $f_{S^{(g)}}(s\mid u)$ and $F_{S^{(g)}}(s\mid u)$ are the density function and cumulative distribution function of $S^{(g)}\mid U=u,$ respectively, and $\mu_{10}(u)=\int \mu_1(s, u)dF_{S^{(0)}}(s\mid u).$ Therefore, 
$$ \sqrt{nh}\left(\begin{array}{c}\widehat{\Delta}^K(\bw)-\Delta^K(\bw) \\ \widehat{\Delta}_{S}^K(\bw)-\Delta_{S}(\bw)\end{array}\right)=\frac{\sqrt{h}}{\sqrt{\pi_1n_1}}\sum_{i=1}^{n_1} \left(\begin{array}{c} \xi_{1i}(u_0(\bw)) \\ \xi_{2i}(u_0(\bw)) \end{array}\right)+\frac{\sqrt{h}}{\sqrt{\pi_0n_0}}\sum_{i=1}^{n_0} \left(\begin{array}{c} \zeta_{1i}(u_0(\bw)) \\ \zeta_{2i}(u_0(\bw)) \end{array}\right)+o_p(1)$$
converges weakly to mean zero bivariate Gaussian distribution with a variance-covariance matrix of $\Sigma_\Delta(\bw).$ It follows from the delta method that 
$$\sqrt{nh}\left\{ \widehat{R}_S(\bw)-R_S(\bw)\right\}$$ also converges weakly to a mean zero Gaussian distribution $N(0, \sigma_R^2(\bw)),$ as the sample size $n \rightarrow \infty.$

\color{black}
\section*{Appendix C}
In Appendix B, we showed that:
$$ \sqrt{nh}\left(\begin{array}{c}\widehat{\Delta}^K(\bw)-\widetilde{\Delta}^K(\bw) \\ \widehat{\Delta}_{S}^K(\bw)-\widetilde{\Delta}_{S}(\bw)\end{array}\right)=o_p(n^{-\epsilon}),$$
implying that 
$$\sqrt{nh}\left\{\widehat{R}_S(\bw)-\widetilde{R}_S(\bw)\right\}=o_p(n^{-\epsilon})$$
where $\widehat{R}_S(\bw)=1-\widetilde{\Delta}_S(\bw)/\widetilde{\Delta}^K(\bw),$ and $\epsilon>0.$ 
Therefore, there exists an $A_n=O(\sqrt{\log(n)})$ and $B_n=O(\sqrt{\log(n)})$ such that
\begin{align*} 
&P\left\{A_n\left( \sup_w \sqrt{nh} \left| \frac{\widehat{R}_S(\bw)-R_S(\bw)}{\sigma_R(\bw)}\right|-B_n\right)\le d\right\}\\
&=P\left\{A_n\left( \sup_w \sqrt{nh} \left| \frac{\widetilde{R}_S(\bw)-R_S(\bw)}{\sigma_R(\bw)}\right|-B_n\right)\le d\right\}\left\{1+o(1)\right\}.
\end{align*}
Under the null hypothesis that $H_0: R_S(\bw)=R_0$, 
$$\sqrt{nh}\left\{|\Omega_{\bw}|^{-1}\int_{\Omega_{\bw}} \widehat{R}_S(\bw)d\bw -R_0\right\}=o_p(1)$$ and, denoting our observed data as $\cal{D}$,
\begin{align*} 
P\left\{A_n\left( {\cal T}-B_n\right)\le d\right\}=&P\left\{A_n\left( \sup_w \sqrt{nh} \left| \frac{\widehat{R}_S(\bw)-R_0}{\sigma_R(\bw)}\right|-B_n\right)\le d\right\}\left\{1+o(1)\right\}\\
&P\left\{A_n\left( \sup_w \sqrt{nh} \left| \frac{\widetilde{R}_S(\bw)-R_0}{\sigma_R(\bw)}\right|-B_n\right)\le d\right\}\left\{1+o(1)\right\}\\
=&P\left\{A_n\left( \sup_{\bw} {\cal T}^{*(b)}-B_n\right)\le d \big |  ~\cal{D}\right\}+o_p(1).
\end{align*}
Thus, the null distribution can be approximated by the described sampling method.

\clearpage 
\section*{Appendix D}
\begin{table}[hptb]
\caption{Estimation results using the proposed parametric and two-stage approaches for AIDS 320 Clinical Trial, for ten $\bw$ combinations, representing the diagonal entries of the grid $\bL$, capturing combinations close to the boundary and the center of the grid. \label{aids_table}}
\begin{center}
\begin{tabular}{|c|c|c|c|c|} \hline
\multicolumn{5}{|c|}{Parametric}\\  \hline
\multicolumn{1}{|c|}{Baseline CD4}&\multicolumn{1}{c|}{Baseline Age}&\multicolumn{1}{c|}{$R_S$}&\multicolumn{1}{c|}{SE($R_S$)}&\multicolumn{1}{c|}{95\% CI for $R_S$}\\ \hline
  17.00      &32.20       &0.32        &0.05        &(0.23, 0.42)\\ 
 31.47      &33.79       &0.30        &0.04        &(0.22, 0.39)\\ 
 45.93      &35.38       &0.28        &0.04        &(0.21, 0.36)\\ 
 60.40      &36.96       &0.26        &0.04        &(0.20, 0.33)\\ 
 74.87      &38.55       &0.25        &0.03        &(0.19, 0.30)\\ 
 89.33      &40.14       &0.23        &0.03        &(0.18, 0.29)\\ 
103.80      &41.72       &0.22        &0.03        &(0.17, 0.27)\\ 
118.27      &43.31       &0.21        &0.03        &(0.16, 0.26)\\ 
132.73      &44.90       &0.20        &0.03        &(0.15, 0.25)\\ 
147.20      &46.48       &0.19        &0.02        &(0.14, 0.24)\\ 
\hline
\multicolumn{5}{|c|}{Two-stage}\\ 
\hline
\multicolumn{1}{|c|}{Baseline CD4}&\multicolumn{1}{c|}{Baseline Age}&\multicolumn{1}{c|}{$R_S$}&\multicolumn{1}{c|}{SE($R_S$)}&\multicolumn{1}{c|}{95\% CI for $R_S$}\\ \hline
  17.00       &32.20        &0.84         &0.21         &(0.36, 1.36) \\ 
 31.47       &33.79        &0.72         &0.20         &(0.22, 0.99) \\ 
 45.93       &35.38        &0.49         &0.22         &(0.17, 0.90) \\ 
 60.40       &36.96        &0.41         &0.14         &(0.21, 0.76) \\ 
 74.87       &38.55        &0.47         &0.12         &(0.26, 0.73) \\ 
 89.33       &40.14        &0.51         &0.17         &(0.19, 0.74) \\ 
103.80       &41.72        &0.43         &0.14         &(0.13, 0.66) \\ 
118.27       &43.31        &0.24         &0.14         &(0.07, 0.57) \\ 
132.73       &44.90        &0.22         &0.14         &(-0.03, 0.54)\\ 
147.20       &46.48        &0.13         &0.13         &(-0.12, 0.41)\\ 
\hline
\end{tabular}
\vspace{3mm}
\end{center}
\end{table}

\end{document}